\documentclass[11pt,reqno]{amsart}

\usepackage{amscd,amsmath,amsfonts,amssymb,amsthm,graphicx,cite,mathrsfs,tikz}

\newtheorem{theorem}{Theorem}[section]
\newtheorem{lemma}{Lemma}[section]
\newtheorem{proposition}{Proposition}[section]
\newtheorem{corollary}{Corollary}[section]

\theoremstyle{remark}
\newtheorem{remark}{Remark}[section]

\renewcommand{\hat}{\widehat}

\newcommand{\nad}[2]{\genfrac{}{}{0pt}{}{#1}{#2}}
\def \vf {\varphi}

\begin{document}

\title{Baker-Akhiezer functions and generalised Macdonald-Mehta integrals}

\author{M.V.~Feigin} \address{School of Mathematics and Statistics, University of Glasgow, 15 University Gardens, Glasgow G12 8QW, United Kingdom}
\email{misha.feigin@glasgow.ac.uk}
\author{M.A.~Halln\"as} \address{Department of Mathematical Sciences, Loughborough University, Loughborough, Leicestershire LE11 3TU, United Kingdom}
\email{M.A.Hallnas@lboro.ac.uk}
\author{A.P.~Veselov} \address{Department of Mathematical Sciences, Loughborough University, Loughborough, Leicestershire LE11 3TU, United Kingdom and Moscow State University, Moscow, Russia}
\email{A.P.Veselov@lboro.ac.uk}

\date{\today}

\begin{abstract}
For the rational Baker-Akhiezer functions associated with special arrangements of hyperplanes with multiplicities we establish an integral identity, which may be viewed as a generalisation of the self-duality property of the usual Gaussian function with respect to the Fourier transformation. We show that the value of properly normalised Baker-Akhiezer function at the origin can be given by an integral of Macdonald-Mehta type and explicitly compute  these integrals for all known Baker-Akhiezer arrangements.
We use the Dotsenko-Fateev integrals to extend this calculation to all deformed root systems, related to the non-exceptional basic classical Lie superalgebras.
\end{abstract}

\maketitle

\section{Introduction}

In 1963 Dyson and Mehta in the paper on the statistical properties of the eigenvalues of random matrices \cite{MD} put forward the following conjecture
$$\int_{\mathbb{R}^n}\prod_{i<j}|(x_i-x_j)|^{2k} d\gamma(x) = \prod_{j=1}^n\frac{\Gamma(1+kj)}{\Gamma(1+k)},$$
where $d\gamma(x) = (2\pi)^{-n/2}e^{-x^2/2}dx,\, x^2=x_1^2+\dots+x_n^2$ is a Gaussian measure and $\Gamma(z)$ is the classical gamma-function (see Conjecture D in \cite{MD}). It became known as  Mehta conjecture, probably because of its later appearance in \cite{M67} and \cite{M74}. Soon after  Bombieri observed that it can be proved using the Selberg integral (see the nice review \cite{FW} for a history of this discovery).

In 1982 Macdonald published his famous list of conjectures \cite{Mac82}, where he generalised Mehta conjecture to any finite Coxeter group $W$:
$$
\int_{\mathbb{R}^n}\prod_{\alpha\in R_+}|(\alpha,x)|^{2k}d\gamma(x) = \prod_{j=1}^n\frac{\Gamma(1+kd_j)}{\Gamma(1+k)},
$$
where $\alpha$ are the normal vectors to the reflection hyperplanes normalised by $(\alpha,\alpha)=2,$ and $d_j$ are the degrees of the generators in $W$-invariant polynomials.

Macdonald checked this for dihedral groups and some other special cases. Opdam \cite{Opd89} computed the integrals corresponding to a Weyl group using so-called shift operators. The remaining exceptional cases $H_3$ and $H_4$ were handled by Opdam \cite{Opd93} using a computer calculation due to Garvan. A uniform proof, valid for all finite Coxeter groups, was later provided by Etingof \cite{Eti10}. In the case when the action of the Coxeter group $W$ on its reflection hyperplanes is not transitive there are two-parameter generalisations of these formulas (see \cite{Mac82} and Section 4 below).

In this paper we present an extension of these results to some special arrangements of hyperplanes with multiplicities, which we call Baker-Akhiezer arrangements.
These arrangements appeared in the theory of rational multidimensional Baker-Akhiezer functions \cite{CFV99} and include Coxeter arrangements as a particular case (see the next Section for the details).

Our motivation came from this theory and related theory of quasi-invariants and $m$-harmonic polynomials, see \cite{FV02, FV03, EG02}. One of our results is an integral representation for the canonical form \cite{FV02} on the space of quasi-invariants $Q^{\mathcal{A}}$ (cf. Macdonald \cite{Mac80}):
$$
	(p,q)^{\mathcal{A}} = \phi(0,0)\int_{i\xi+\mathbb{R}^n}\frac{\left(e^{L/2}p\right)(-ix)\left(e^{L/2}q\right)(ix)}{A_m(x)^2}d\gamma(x),\quad p,q\in Q^{\mathcal{A}},
$$
where
$A_m(x) = \prod_{\alpha\in\mathcal{A}}(\alpha,x)^{m_\alpha}$ is the product of the defining forms of the arrangement with multiplicities, $\phi(0,0)$ is the value of the corresponding Baker-Akhiezer function at zero and 
$$
L = \Delta - \sum_{\alpha\in\mathcal{A}} 2m_{\alpha}(\alpha,x)^{-1}\partial_\alpha.
$$
Because of the zeroes in the denominator we need a shift of the integration contour to the complex domain (cf. Etingof and Varchenko \cite{EV05}, Grinevich and Novikov \cite{GN09, GN11}, Chalykh and Etingof \cite{CE11}).
Since $(1,1)^{\mathcal{A}} =1$ we see that (the inverse of) this value itself can be given by the integral of Macdonald-Mehta type:
$$\phi(0,0)^{-1} = \int_{i\xi+\mathbb{R}^n}\frac{d\gamma(x)}{A_m(x)^2}.$$
We explicitly compute this value for all known Baker-Akhiezer arrangements. In the two-dimensional case we use recent results by Berest et al \cite{BCE08} and by Johnston and one of the authors \cite{FJ12}.
For the multidimensional non-Coxeter Baker-Akhiezer arrangements we follow Bombieri's calculation using a version of the Selberg integral found by Dotsenko and Fateev \cite{DF}.

In fact, we use the Dotsenko-Fateev integral to evaluate generalised Macdonald-Mehta integrals for all deformed root systems, corresponding to the non-exceptional basic classical Lie superalgebras \cite{SV}. The corresponding deformed Calogero-Moser operator
$$
L_{\mathcal A}=-\Delta+\sum_{\alpha \in \mathcal A}\frac{m_\alpha(m_\alpha+1)(\alpha,\alpha)}{(\alpha,x)^2}+x^2
$$
has a "pseudo-ground state" of the form
$
\psi_0=\prod_{\alpha \in \mathcal A}(\alpha,x)^{-m_\alpha}\exp(-\frac{1}{2}x^2),
$
which is always singular since the multiplicities of the isotropic roots are 1. One can define its "norm" by the integral
$$
I_{\mathcal A}=\int_{i\xi+\mathbb{R}^n} dx \prod_{\alpha \in \mathcal A}(\alpha,x)^{-2m_\alpha}e^{-x^2}.
$$
In particular, in the deformed $BC(n,m)$ case, corresponding to orthosymplectic Lie superalgebras, we have
the following generalised Macdonald-Mehta integral defined by
$$
I(\alpha, \rho)=\int_{{\mathbb R}^n+i \xi} dt \int_{{\mathbb R}^m+i\eta} d\tau  \frac{  \prod_{i=1}^n t_i^{-2\frac{\alpha}{\rho}+1} \prod_{j=1}^m \tau_j^{2\alpha+1} e^{-\frac12\sum_{j=1}^m\tau_j^2-\frac{1}{2}\sum_{i=1}^n t_i^2}}{\prod_{i<j}^n (t_i^2 - t_j^2)^{-\frac{2}{\rho}}  \prod_{i<j}^m (\tau_i^2- \tau_j^2)^{-2\rho}   \prod_{j=1}^m \prod_{i=1}^n (\rho t_i^2 +  \tau_j^2)^2},
$$
where $\rho<0$, $\alpha \in \mathbb R$, and  $\xi_n>\ldots>\xi_1>\eta_m>\ldots>\eta_1>0$.
Using Dotsenko-Fateev formula we show (see Section 6 below) that
\begin{equation*}
\begin{split}
I(\alpha,\rho) &= C\prod_{k=1}^m \frac{\Gamma(1-\rho)}{\Gamma(1-k \rho) \prod_{j=1}^n (k\rho-j)}
\prod_{k=1}^n \frac{\Gamma(1-\rho^{-1})}{\Gamma(1-k \rho^{-1})}\\ &\quad \times \prod_{j=0}^{n-1}\frac{1}{\Gamma(\frac{\alpha-j}{\rho})} \prod_{j=0}^{m-1} \frac{1}{\Gamma(-\alpha-j \rho)} \prod_{j=0}^{m-1}\prod_{k=0}^{n-1} \frac{1}{\alpha+j \rho -k}.
\end{split}
\end{equation*}
with
$$
C=(2\pi)^{m+n} 2^{2(m+n)-2m n -\frac{n\alpha}{\rho}+\frac{n(n-1)}{\rho}+m\alpha+m(m-1)\rho} e^{\pi i (\frac{m+n}{2}+m\alpha-\frac{n\alpha}{\rho})}.$$

Another result of this paper is an integral identity for the Baker-Akhiezer function, which may be viewed as a generalisation of the self-duality property of the usual Gaussian with respect to the Fourier transformation:
$$
	\int_{i\xi+\mathbb{R}^n}\frac{\phi(-ix,\lambda)\phi(ix,\mu)}{A_m(x)^2}d\gamma(x) = e^{-(\lambda^2+\mu^2)/2}\phi(\lambda,\mu).
$$
In the crystallographic case  the difference version of this identity was recently established by Chalykh and Etingof \cite{CE11}.

\section{Multidimensional Baker-Akhiezer functions and integral identity}

The notion of rational Baker-Akiezer function related to an arrangement of hyperplanes with multiplicities was introduced by Chalykh and one of the authors in \cite{CV90} (see also \cite{VSC93}) as a multidimensional version of Krichever's axiomatic \cite{Kri}.

We recall that there are two different axiomatics for rational Baker-Akhiezer functions: the original axiomatics, proposed in \cite{CV90}, was later revised in \cite{CFV99} to cover more cases. For our purposes it is important to make use of the original, and more restrictive, axiomatics. In addition, we shall restrict our attention to hyperplanes in a real (rather than complex) Euclidean space $\mathbb{R}^n$ equipped with the standard inner product $(\cdot,\cdot)$.

Let $\mathcal{A}$ be a finite collection of noncolinear vectors $\alpha=(\alpha_1,\ldots,\alpha_n)\in\mathbb{R}^n$ with multiplicities $m_\alpha\in\mathbb{N}$. To each vector $\alpha\in\mathcal{A}$ there corresponds a hyperplane
\begin{equation}
	\Pi_\alpha=\lbrace x\in\mathbb{R}^n: (x,\alpha)=0\rbrace.
\end{equation}
These hyperplanes partition $\mathbb{R}^n$ into a finite number of regions. We shall refer to the connected components of $\mathbb{R}^n\setminus\cup_{\alpha\in\mathcal{A}}\Pi_\alpha$ as the (open) chambers of $\mathbb{R}^n$. Since any two vectors $\alpha$ and $-\alpha$ determine the same hyperplane, we may and will assume that there exists a vector $v\in\mathbb{R}^n$ such that $(\alpha,v)>0$ for all $\alpha\in\mathcal{A}$. In particular, this allows us to define the 'negative' chamber of $\mathbb{R}^n$ as the set of points $x\in\mathbb{R}^n$ such that $(\alpha,x)<0$ for all $\alpha\in\mathcal{A}$. In comparison with \cite{CV90} and \cite{CFV99} we shall use a gauge that differs by the factor
\begin{equation}
	A_m(x) = \prod_{\alpha\in\mathcal{A}}(x,\alpha)^{m_\alpha}.
\end{equation}

Following \cite{CV90} a function $\phi(x,\lambda)$, $x,\lambda\in\mathbb{C}^n$, will be called a Baker-Akhiezer function associated to the collection $\mathcal{A}$ if it satisfies the following two conditions:
\begin{enumerate}
\item $\phi(x,\lambda)$ is of the form
\begin{equation}\label{BAform}
	\phi(x,\lambda) = P(x,\lambda)e^{(x,\lambda)}
\end{equation}
with $P(x,\lambda)$ a polynomial in $x$ with highest degree term equal to $A_m(x)A_m(\lambda)$;
\item for all $\alpha\in\mathcal{A}$,
\begin{equation}\label{quasiinv}
	\partial_\alpha\phi(x,\lambda) = \partial_\alpha^3\phi(x,\lambda) =\cdots= \partial_\alpha^{2m_\alpha-1}\phi(x,\lambda)\equiv 0,\quad x\in\Pi_\alpha,
\end{equation}
where $\partial_\alpha = (\alpha,\partial/\partial x)$ is the normal derivative corresponding to the vector $\alpha$.
\end{enumerate}

In order for $\phi(x,\lambda)$ to be symmetric under the interchange of $x$ and $\lambda$ (see below), we have chosen the normalisation of $\phi$ somewhat differently from that in \cite{CV90}. In \cite{CFV99} the definition above was revised as follows: a function of the form
\begin{equation}\label{revBAform}
	\psi(x,\lambda) = \frac{P(x,\lambda)}{A_m(x)A_m(\lambda)}e^{(x,\lambda)}
\end{equation}
was considered, where, as above, $P(x,\lambda)$ is a polynomial in $x$ with highest degree term $A_m(x)A_m(\lambda)$, and with the property that
\begin{equation}\label{revquasiinv}
\begin{split}
	\partial_\alpha\big(\psi(x,\lambda)(x,\alpha)^{m_\alpha}\big) &= \partial_\alpha^3\big(\psi(x,\lambda)(x,\alpha)^{m_\alpha}\big)\\ &=\cdots= \partial_\alpha^{2m_\alpha-1}\big(\psi(x,\lambda)(x,\alpha)^{m_\alpha}\big)\equiv 0,\quad x\in\Pi_\alpha.
\end{split}
\end{equation}
Moreover, it was shown that if there exists a Baker-Akhiezer function $\phi$, satisfying \eqref{BAform}--\eqref{quasiinv}, then $\psi(x,\lambda)=\phi(x,\lambda)/(A_m(x)A_m(\lambda))$ satisfies conditions \eqref{revBAform}--\eqref{revquasiinv}, see Corollary 2.7 in \cite{CFV99}. However, the converse statement does not hold true, i.e., there are collections of vectors $\mathcal{A}$ such that $\psi$ exists but $\phi$ does not.

The existence of $\phi$ puts a strong restriction on the arrangement. The known cases besides Coxeter arrangements include their deformed  versions ${\mathcal A}_{m}(p)$ and ${\mathcal C}_{m+1}(r,s)$ as well as some special 2-dimensional configurations, see Section \ref{MMsec}-\ref{SectDeform} below.

It is known that if a Baker-Akhiezer function $\phi$ exists, then it is unique and symmetric with respect to interchange of $x$ and $\lambda$:
\begin{equation*}
	\phi(x,\lambda) = \phi(\lambda,x),
\end{equation*}
see Theorem 2.3 in \cite{CFV99} and note the difference in gauge. Moreover, it satisfies the algebraically integrable differential equation
\begin{equation}\label{CMeq}
	L\phi = \lambda^2\phi
\end{equation}
with
\begin{equation}\label{CM}
	L = \Delta - \sum_{\alpha\in\mathcal{A}}\frac{2m_\alpha}{(\alpha,x)}\partial_\alpha,
\end{equation}
where $\Delta$ denotes the ordinary Laplace operator. In terms of the operator $L$ the Baker-Akhiezer function is given by Berest's formula
\begin{equation}\label{Berest}
	\phi(x,\lambda) = \big(2^{|m|} |m|!\big)^{-1}\big(L-\lambda^2\big)^{|m|}\left(A_m(x)^2e^{(\lambda,x)}\right),
\end{equation}
where $|m|=\sum_{\alpha\in\mathcal{A}}m_\alpha$, see Theorem 3.1 in \cite{CFV99}.

\subsection{Integral identity}
We are now ready to proceed to state the integral identity announced above, which forms the main result of this section. Given $\xi\in\mathbb{R}^n$, we let $\gamma$ denote the Gaussian measure on $i\xi+\mathbb{R}^n$ given by
\begin{equation*}
	d\gamma(x) = (2\pi)^{-n/2}e^{-x^2/2}dx.
\end{equation*}

\begin{theorem}\label{intThm}
For any $\lambda,\mu\in \mathbb{C}^n$ and any $\xi\in \mathbb{R}^n\setminus\cup_{\alpha\in\mathcal{A}}\Pi_\alpha$, we have
\begin{equation}\label{int}
	\int_{i\xi+\mathbb{R}^n}\frac{\phi(-ix,\lambda)\phi(ix,\mu)}{A_m(x)^2}d\gamma(x) = e^{-(\lambda^2+\mu^2)/2}\phi(\lambda,\mu).
\end{equation}
\end{theorem}

Before continuing to the proof of the theorem we note that a direct consequence is an integral expression for the value of the Baker-Akhiezer function at the origin. Indeed, suppose that $\phi(0,0)\neq 0$. It can be directly inferred from Berest's formula that
\begin{equation*}
	\phi(x,\lambda) = \big(P_{|m|}(x,\lambda)+\cdots+P_j(x,\lambda)+\cdots+P_0(x,\lambda)\big)e^{(x,\lambda)}
\end{equation*}
for some polynomials $P_j(x,\lambda)$ that are homogeneous of degree $j$ in both $x$ and $\lambda$. It follows that $\phi(x,0)=\phi(0,\lambda)=\phi(0,0)$. Setting $\lambda=\mu=0$ in Theorem \ref{intThm} we thus arrive at the following Corollary.

\begin{corollary}\label{phiCor}
Assume that $\phi(0,0)\neq 0$, and let $\xi$ be as in Theorem \ref{intThm}. Then, we have
\begin{equation*}
	\phi(0,0) = \left(\int_{i\xi+\mathbb{R}^n}\frac{d\gamma(x)}{A_m(x)^2}\right)^{-1}.
\end{equation*}
\end{corollary}

For all known cases of existence of the Baker-Akhiezer function the value $\phi(0,0)\neq 0$, see \cite{EG02, FV03}. We shall consider these cases in more detail in Section \ref{MMsec}-\ref{SectDeform}, and explicitly compute the corresponding value of $\phi(0,0)$.

\subsection{Proof of Theorem \ref{intThm}}
Let $I_\xi$ denote the integral in the left-hand side of \eqref{int}. It is clear that $I_\xi$ is well-defined for any choice of $\xi\in \mathbb{R}^n\setminus\cup_{\alpha\in\mathcal{A}}\Pi_\alpha$. A key ingredient in the proof of the theorem is the fact that $I_\xi$ does not depend on this choice.

\begin{lemma}\label{indepLemma}
The integral $I_\xi$ is independent of the value of $\xi$ provided it remains in $\mathbb{R}^n\setminus\cup_{\alpha\in\mathcal{A}}\Pi_\alpha$.
\end{lemma}

\begin{proof}
By Cauchy's theorem, the integral $I_\xi$ does not change when $\xi$ varies within a given chamber. It is thus sufficient to show that for any two adjacent chambers $C_+$ and $C_-$ we have $I_{\xi_+}=I_{\xi_-}$ for some $\xi_\pm\in C_\pm$. Fix two such chambers, and suppose that they are separated by the hyperplane $\Pi_\alpha$ for some $\alpha\in\mathcal{A}$. We may assume that $\xi_\pm = p\pm\eta\alpha/|\alpha|$ for some $p\in\Pi_\alpha$ and non-zero real number $\eta$.

We proceed to evaluate the difference $I_{\xi_+}-I_{\xi_-}$. To this end, we shall make use of the map $\Pi_\alpha\times\mathbb{R}\to\mathbb{R}^n$ given by $(y,z)\mapsto y+z\alpha/|\alpha|$. It is clear that the Jacobian determinant of this map equals 1. For convenience, we denote the integrand of $I_\xi$ by $F$, i.e.,
\begin{equation*}
	F(x) = \frac{\phi(-ix,\lambda)\phi(ix,\mu)}{A_m(x)^2}\frac{e^{-x^2/2}}{(2\pi)^{n/2}}.
\end{equation*}
Then, we have
\begin{multline*}
	I_{\xi_+} - I_{\xi_-} = \int_{\Pi_\alpha}\left(\int_{\mathbb{R}}F\left(i\xi_+ +y+z\frac{\alpha}{|\alpha|}\right)dz-\int_{\mathbb{R}}F\left(i\xi_- +y+z\frac{\alpha}{|\alpha|}\right)dz\right)dy\\  =\int_{\Pi_\alpha}\Bigg(\int_{\mathbb{R}}F\left(ip +y+(i\eta +z)\frac{\alpha}{|\alpha|}\right)dz\\ -\int_{\mathbb{R}}F\left(ip+y+(-i\eta+z)\frac{\alpha}{|\alpha|}\right)dz\Bigg)dy.
\end{multline*}

We note that the function $G_{ip+y}:w\mapsto F(ip+y+w\alpha/|\alpha|)$ on $\mathbb{C}$ is meromorphic with a pole of order $2m_\alpha$ located at $w=0$. It follows from the residue theorem that the difference between the two inner integrals is proportional to
\begin{equation*}
	\big((2m_\alpha-1)!\big)\mathrm{Res}_{w=0}\big(G_{ip+y}(w)\big) = \partial_\alpha^{2m_\alpha-1}\left(\frac{\phi(-ix,\lambda)\phi(ix,\mu)e^{-x^2/2}}{\prod_{\substack{\beta\in\mathcal{A}\\ \beta\neq\alpha}}(\beta,x)^{2m_\beta}}\right),\quad x=ip+y.
\end{equation*}
With $\psi(x,\lambda)=\phi(x,\lambda)/A_m(x)$, this residue takes the form
\begin{equation*}
	 \partial_\alpha^{2m_\alpha-1}\Big(\big\lbrack\psi(-ix,\lambda)(-ix,\alpha)^{m_\alpha}\big\rbrack\big\lbrack\psi(ix,\lambda)(ix,\alpha)^{m_\alpha}\big\rbrack e^{-x^2/2}\Big).
\end{equation*}
We note that each factor in the above expression satisfies the condition \eqref{quasiinv}. For $e^{-x^2/2}$ this is immediate from invariance under reflections, and for $\psi(\pm ix,\lambda)(\pm ix,\alpha)^{m_\alpha}$ this follows from Corollary 2.7 in \cite{CFV99}, see the remark succeeding \eqref{revquasiinv}. Since the condition \eqref{quasiinv} is closed under multiplication, the residue vanishes for $x=ip+y$, $p,y\in\Pi_\alpha$.
\end{proof}

Since differentiation under the integral is allowed in \eqref{int}, it is clear that $I_\xi(\lambda,\mu)$ is an entire function in both $\lambda$ and $\mu$. It is also clear that it satisfies the conditions \eqref{quasiinv} (with $(\lambda,\mu)$ substituted for $(x,\lambda)$). By uniqueness of the Baker-Akhiezer function, it is thus sufficient to verify that $I_\xi(\lambda,\mu)e^{(\lambda^2+\mu^2)/2}$ is of the form \eqref{BAform}.

Due to Lemma \ref{indepLemma} we may assume that $\xi$ is contained in the 'negative' chamber of $\mathbb{R}^n$, i.e., that $(\alpha,\xi)<0$ for all $\alpha\in\mathcal{A}$. We let $\mathcal{M}$ denote the set of multiplicities $n=\lbrace n_\alpha\rbrace_{\alpha\in\mathcal{A}}$ such that $n_\alpha=0,\ldots,m_\alpha$ for all $\alpha\in\mathcal{A}$. From Berest's formula \eqref{Berest} and the form of the Calogero-Moser operator \eqref{CM} it follows that
\begin{equation}
	\phi(x,\lambda) = e^{(x,\lambda)}\sum_{n\in\mathcal{M}}Q_n(\lambda)\prod_{\alpha\in\mathcal{A}}(\alpha,x)^{n_\alpha}
\end{equation}
for some polynomials $Q_n$ of degree $|n|\equiv\sum_{\alpha\in\mathcal{A}}n_\alpha$. In particular, we have $Q_m(\lambda) = A_m(\lambda)$. Inserting this expansion into the left-hand side of \eqref{int}, we deduce
\begin{equation}\label{Iexp}
	I_\xi(\lambda,\mu)e^{(\lambda^2+\mu^2)/2}e^{-(\lambda,\mu)} = \sum_{n,n^\prime\in\mathcal{M}}Q_n(\lambda)Q_{n^\prime}(\mu)(2\pi)^{-n/2}\int_{i\xi+\mathbb{R}^n}\frac{e^{-(x+i\lambda-i\mu)^2/2}dx}{\prod_{\alpha\in\mathcal{A}}(\alpha,x)^{2m_\alpha-n_\alpha-n^\prime_\alpha}}.
\end{equation}
For $n=n^\prime=m$ we have the term
\begin{equation*}
\begin{split}
	A_m(\lambda)A_m(\mu)\int_{i\xi+\mathbb{R}^n}\frac{e^{-(x+i\lambda-i\mu)e^2/2}dx}{(2\pi)^{n/2}} &= A_m(\lambda)A_m(\mu)\int_{\mathbb{R}^n}\frac{e^{-x^2/2}dx}{(2\pi)^{n/2}}\\ &= A_m(\lambda)A_m(\mu).
\end{split}
\end{equation*}

We proceed to show that the remaining terms are bounded by a polynomial of degree at most $|m|-1$. To this end, we observe that, for $\lambda\in\mathbb{C}^n$ such that $(\alpha,{\rm Re}\, \lambda)\geq 0$ for all $\alpha\in\mathcal{A}$, we can shift the domain of integration in all integrals in \eqref{Iexp} according to $\xi\to\xi-{\rm Re}\, \lambda$ without crossing any of the poles of the corresponding integrand. (Recall that $\xi$ is assumed to be contained in the 'negative' chamber.) Doing so, we rewrite them as follows:
\begin{equation*}
	\int_{i\xi+\mathbb{R}^n}\frac{e^{-(x-i\mu)^2/2}dx}{\prod_{\alpha\in\mathcal{A}}(\alpha,x-i\lambda)^{2m_\alpha-n_\alpha-n^\prime_\alpha}}.
\end{equation*}
For our purposes, a sufficient bound for these latter integrals is given in the Lemma below.

\begin{lemma}
Let $k\in\mathcal{M}$. Fix $\xi\in\mathbb{R}^n$ such that $(\alpha,\xi)<0$ for all $\alpha\in\mathcal{A}$, and fix $\mu\in\mathbb{C}^n$. Then, there exists a positive number $C$ such that
\begin{equation*}
	\left|\int_{i\xi+\mathbb{R}^n}\frac{e^{-(x-i\mu)^2/2}dx}{\prod_{\alpha\in\mathcal{A}}(\alpha,x-i\lambda)^{k_\alpha}}\right| < \frac{C}{\prod_{\alpha\in\mathcal{A}}\big|(\alpha,\lambda)\big|^{k_\alpha}},\quad \big(\alpha,{\rm Re}\, \lambda\big)\geq 0~\forall \alpha\in\mathcal{A}.
\end{equation*}
\end{lemma}

\begin{proof}
For convenience, we introduce the short-hand notation
\begin{equation*}
	\binom{k}{\ell}\equiv \prod_{\alpha\in\mathcal{A}}\binom{k_\alpha}{\ell_\alpha},\quad k,\ell\in\mathcal{M}.
\end{equation*}
Then, we have
\begin{multline}\label{eqs}
	 \prod_{\alpha\in\mathcal{A}}(\alpha,i\lambda)^{k_\alpha}\int_{i\xi+\mathbb{R}^n}\frac{e^{-(x-i\mu)^2/2}dx}{\prod_{\alpha\in\mathcal{A}}(\alpha,x-i\lambda)^{k_\alpha}}\\ = \prod_{\alpha\in\mathcal{A}}(\alpha,x-(x-i\lambda))^{k_\alpha}\int_{i\xi+\mathbb{R}^n}\frac{e^{-(x-i\mu)^2/2}dx}{\prod_{\alpha\in\mathcal{A}}(\alpha,x-i\lambda)^{k_\alpha}}\\ = \sum_{\substack{\ell\in\mathcal{M}\\ \ell_\alpha\leq k_\alpha}}(-1)^{|\ell|}\binom{k}{\ell}\int_{i\xi+\mathbb{R}^n}e^{-(x-i\mu)^2/2}\prod_{\alpha\in\mathcal{A}}\left(\frac{(\alpha,x)}{(\alpha,x-i\lambda)}\right)^{k_\alpha-\ell_\alpha}dx,
\end{multline}
where $|\ell|=\sum_{\alpha\in\mathcal{A}}\ell_\alpha$. We note that, for $x\in \mathbb{R}^n+i\xi$,
\begin{equation*}
	\big|(\alpha,x-i\lambda)\big|>\big|(\alpha,\xi-{\rm Re}\, \lambda)\big|.
\end{equation*}
Consequently, as a function of $\lambda$, each integral in the last line of \eqref{eqs} is a bounded function on the subset of $\mathbb{C}^n$ given by $(\alpha,{\rm Re}\, \lambda)\geq 0$ for all $\alpha\in\mathcal{A}$.
\end{proof}

Since the $Q_n(\lambda)$ are polynomials of degree $|n|$ with $Q_m(\lambda)=A_m(\lambda)$, we can thus conclude that
\begin{equation*}
	I_\xi(\lambda,\mu)e^{(\lambda^2+\mu^2)/2} = e^{(\lambda,\mu)}\big(A_m(\lambda)A_m(\mu)+R(\lambda,\mu)\big)
\end{equation*}
for some function $R(\lambda,\mu)$ that is entire in $\lambda$, and that there exists a (positive) function $C(\mu)$ such that
\begin{equation*}
	|R(\lambda,\mu)| < C(\mu)\sum_{\substack{\alpha\in\mathbb{N}^n\\ |\alpha|\leq |m|-1}}|\lambda_1|^{\alpha_1}\cdots|\lambda_n|^{\alpha_n},\quad \big(\alpha,{\rm Re}\, \lambda\big)\geq 0~\forall \alpha\in\mathcal{A}.
\end{equation*}
Moreover, by moving $\xi$ to the other chambers of $\mathbb{R}^n$ we extend this bound to all $\lambda\in\mathbb{C}^n$. Since $R$ is an entire function, it is in fact a polynomial of degree at most $|m|-1$, and we have thus verified that $I_\xi(\lambda,\mu)e^{(\lambda^2+\mu^2)/2}$ is indeed of the form \eqref{BAform}. This concludes the proof of Theorem \ref{intThm}.

\section{A bilinear form on quasi-invariants}
We proceed to deduce a further consequence of Theorem \ref{intThm}, related to a natural bilinear form on the algebra of so-called quasi-invariants. Let us therefore fix a collection $\mathcal{A}$  of vectors $\alpha \in {\mathbb R}^n$ with multiplicities $m_\alpha\in \mathbb N$ such that the corresponding Baker-Akhiezer function $\phi$ exists. A (real) polynomial $p$ on $\mathbb{R}^n$ is said to be quasi-invariant if it satisfies \eqref{quasiinv} for all $\alpha\in\mathcal{A}$. It is readily seen that the set of such polynomials forms a ring, which we shall denote by $Q^{\mathcal{A}}$. Rings of such quasi-invariants appeared first in the work of Chalykh and one of the authors \cite{CV90} in the context of quantum Calogero-Moser systems.

We recall that to any $p\in Q^{\mathcal{A}}$ one can associate a differential operator $L_p$ by requiring that
\begin{equation}\label{hoCMeq}
	L_p\phi(x,\lambda)=p(\lambda)\phi(x,\lambda).
\end{equation}
Moreover, the set of such operators forms a commutative ring isomorphic to $Q^{\mathcal{A}}$, and the differential operator \eqref{CM} corresponds to the polynomial $x^2$, cf.~\eqref{CMeq}. For further details see, e.g., \cite{CFV99}.

On the algebra of quasi-invariants there is a natural bilinear form, given by
\begin{equation}\label{bilindef}
	(p,q)^{\mathcal{A}} = (L_pq)(0),\quad p,q\in Q^{\mathcal{A}}.
\end{equation}
As a straightforward consequence of Theorem \ref{intThm}, we can establish an integral representation of this bilinear form. It will become clear below that the assumption $\phi(0,0)\neq 0$ will again be essential.

Let $p,q\in Q^{\mathcal{A}}$. We now proceed in three steps to deduce the desired integral representation from \eqref{int}. First, using \eqref{CMeq}, we rewrite this equation as
\begin{equation*}
	\int_{i\xi+\mathbb{R}^n}\frac{\left(e^{-L/2}\phi(-ix,\lambda)\right)\left(e^{-L/2}\phi(ix,\mu)\right)}{A_m(x)^2}d\gamma(x) = \phi(\lambda,\mu),
\end{equation*}
where $L$ is the operator \eqref{CM}.
Second, acting by $L_q$ in the variables $\mu$ and making use of the fact $L_q$ commutes with $L$, and then setting $\mu=0$, we infer from \eqref{hoCMeq} and the assumption $\phi(0,0)\neq 0$ that
\begin{equation*}
	\int_{i\xi+\mathbb{R}^n}\frac{\left(e^{-L/2}\phi(-ix,\lambda)\right)\left(e^{-L/2}q(ix)\right)}{A_m(x)^2}d\gamma(x) = q(\lambda).
\end{equation*}
Third, acting by $L_p$ and setting $\lambda=0$, we similarly deduce
\begin{equation*}
	\phi(0,0)\int_{i\xi+\mathbb{R}^n}\frac{\left(e^{-L/2}p(-ix)\right)\left(e^{-L/2}q(ix)\right)}{A_m(x)^2}d\gamma(x) = (L_pq)(0).
\end{equation*}
We thus arrive at the following Theorem (cf. Equation (6) in Macdonald \cite{Mac80}):
\begin{theorem}\label{formThm}
Assume that $\phi(0,0)\neq 0$. Then we have
\begin{equation}
	(p,q)^{\mathcal{A}} = \phi(0,0)\int_{i\xi+\mathbb{R}^n}\frac{\left(e^{L/2}p\right)(-ix)\left(e^{L/2}q\right)(ix)}{A_m(x)^2}d\gamma(x),\quad p,q\in Q^{\mathcal{A}}.
	\end{equation}
\end{theorem}

In particular, if $p$ and $q$ are {\it harmonic} in the sense that $Lp=Lq=0,$ then
\begin{equation}
	(p,q)^{\mathcal{A}} = \phi(0,0)\int_{i\xi+\mathbb{R}^n}\frac{p(-ix)q(ix)}{A_m(x)^2}d\gamma(x).
\end{equation}

\section{Macdonald-Mehta integrals -- Coxeter cases}\label{MMsec}

We recall that Corollary \ref{phiCor} yields an integral expression for the value of $\phi(0,0)$. The main purpose of this section is to compute these integrals when the configuration of vectors $\mathcal{A}=R_+$ with $R_+$ positive roots of a (finite) Coxeter group. We shall accomplish this by determining the analytic continuations of known evaluations of integrals of so-called Macdonald-Mehta type.

More specifically, we shall first consider the case of equal multiplicities, and deduce a uniform formula, valid for all Coxeter groups. In case the action on its root system consists of two orbits there are expressions for the corresponding Macdonald-Mehta integrals allowing for two distinct parameters, each associated to one of the two orbits. The latter cases consists of the groups given by the root systems $B_n$ (the root systems $C_n$ yield essentially the same integrals), $F_4$, as well as the dihedral groups $I_2(2m)$. Apart from the dihedral groups, which are contained in the discussion in Section \ref{2DMMSec}, these cases are treated below in Section \ref{2pSec}.

\subsection{Equal multiplicities}
Let $W$ be a finite Coxeter group, i.e., a finite group generated by orthogonal  reflections with respect to hyperplanes in a Euclidean space $V$ of some dimension $n$. We shall identify $V$ with $\mathbb{R}^n$, equipped with the standard positive definite symmetric bilinear form $(\cdot,\cdot)$. We fix a corresponding choice of positive roots $R_+$ such that $(\alpha,v)>0$ for all $\alpha\in R_+$ and some $v\in\mathbb{R}^n$, and we normalise all the roots by
\begin{equation*}
	(\alpha,\alpha)=2,\quad \alpha\in R_+.
\end{equation*}
In this section we shall consider only the case of equal multiplicities, i.e., we shall assume that $m_\alpha=m$ for all $\alpha\in R_+$ and for some $m\in\mathbb{N}$. This entails drastic simplifications for many of the formulae involved.

We continue by recalling the so-called Macdonald-Mehta integral associated to $W$. For that, we let $d_j$, $j=1,\ldots,n$, denote the degrees of $n$ homogeneous generators of the sub-algebra of $\mathbb{R}\lbrack x_1,\ldots,x_n\rbrack$ consisting of all $W$-invariant polynomials, see, e.g., Chapter 3 in Humphreys \cite{Hum90}.

\begin{theorem}\label{MMThm}
For ${\rm Re}\, k\geq 0$, one has
\begin{equation}\label{MMint}
	\int_{\mathbb{R}^n}\prod_{\alpha\in R_+}|(\alpha,x)|^{2k}d\gamma(x) = \prod_{j=1}^n\frac{\Gamma(1+kd_j)}{\Gamma(1+k)}.
\end{equation}
\end{theorem}

As we shall see below, the value of $\phi(0,0)$, with $\phi$ the Baker-Akhiezer function associated to the configuration $\mathcal{A}\equiv R_+$, can be deduced from the Macdonald-Mehta integral \eqref{MMint}. To this end, we shall first determine the effect of replacing each factor $|(\alpha,x)|$, $\alpha\in R_+$, by $(\alpha,x)$ in the left-hand side of \eqref{MMint}.

Let $C_+$ denote the 'positive' chamber of $\mathbb{R}^n$, given by $(\alpha,x)>0$ for all $\alpha\in R_+$. Fixing $\xi\in C_+$, the (real) hyperplane $it\xi+\mathbb{R}^n$, $t>0$, does not intersect any of the reflection hyperplanes $(\alpha,x)=0$, $\alpha\in R_+$. It follows that the product $\prod_{\alpha\in R_+}(\alpha,x)^{2k}$ has a (unique) continuous branch on $it\xi+\mathbb{R}^n$ that tends to the positive branch on $C_+$ as $t\downarrow 0$. For this branch, we have the following Lemma.

\begin{lemma}
For $\xi\in C_+$ and ${\rm Re}\, k\geq 0$, one has
\begin{equation}\label{anCont}
	\int_{i\xi+\mathbb{R}^n}\prod_{\alpha\in R_+}(\alpha,x)^{2k}d\gamma(x) = \frac{1}{|W|}\left(\prod_{j=1}^n\frac{1-e^{2\pi ikd_j}}{1-e^{2\pi ik}}\right)\int_{\mathbb{R}^n}\prod_{\alpha\in R_+}|(\alpha,x)|^{2k}d\gamma(x).
\end{equation}
\end{lemma}

\begin{remark}
The statement can be readily extracted from the proof of Proposition 4.8 in Etingof \cite{Eti10}, see also the proof of Lemma \ref{contLemma} below.
\end{remark}

The left-hand side of \eqref{anCont} is clearly an entire function in $k$. On the other hand, using the reflection equation $\Gamma(z)\Gamma(1-z)=\pi/\sin\pi z$ and Theorem \ref{MMThm}, we find that the right-hand side is given by
\begin{equation*}
	\frac{1}{|W|}\prod_{j=1}^ne^{\pi ik(d_j-1)}\frac{\Gamma(-k)}{\Gamma(-kd_j)},
\end{equation*}
which yields its analytic continuation to ${\rm Re}\, k<0$. In this equation we may set $k=-m$, $m\in\mathbb{N}$, and thus arrive at the desired evaluation of $\phi(0,0)$, cf.~Corollary \ref{phiCor}. The resulting formula can be simplified somewhat by recalling that
\begin{equation*}
	\sum_{j=1}^n(d_j-1) = |R_+|,
\end{equation*}
see, e.g., Section 3.9 in Humphreys \cite{Hum90}.

\begin{proposition}\label{intProp}
Let $\xi\in\mathbb{R}^n$ be such that $(\alpha,\xi)\neq 0$ for all $\alpha\in R_+$, and let $m\in\mathbb{N}$. Then, we have
\begin{equation*}
	\int_{i\xi+\mathbb{R}^n}\frac{d\gamma(x)}{\prod_{\alpha\in R_+}(\alpha,x)^{2m}}dx = \frac{(-1)^{m|R_+|}}{|W|}\prod_{j=1}^n\frac{\Gamma(m)}{\Gamma(md_j)}.
\end{equation*}
\end{proposition}

\begin{remark}
The fact that we need not require that $\xi$ is contained in the 'positive' chamber $C_+$ can be directly inferred from Lemma \ref{indepLemma} by setting $\lambda=\mu=0$ in the integral $I_\xi$.
\end{remark}

We note that the integral in Proposition \ref{intProp} is related to the Macdonald-Mehta integral in Theorem \ref{MMThm}, for $k=m$, by a simple 'reflection' equation.

\begin{corollary}
Let $\xi$ be as in Proposition \ref{intProp}. For $m\in\mathbb{Z}$, let
\begin{equation*}
	G_W(m) = \int_{i\xi+\mathbb{R}^n}\prod_{\alpha\in R_+}(\alpha,x)^{2m}d\gamma(x).
\end{equation*}
Then, we have
\begin{equation*}
	G_W(m)G_W(-m) = (-1)^{m|R_+|}.
\end{equation*}
\end{corollary}

\begin{proof}
The statement is a direct consequence of Proposition \ref{intProp}, Theorem \ref{MMThm}, and the fact that
\begin{equation}\label{degprod}
	\prod_{j=1}^nd_j=|W|
\end{equation}
(see, e.g., Section 3.9 in Humphreys \cite{Hum90} for its proof).
\end{proof}

We conclude this section by briefly discussing the connection to the theory of so-called $m$-harmonic polynomials. The space $H_m$ of such polynomials was introduced in \cite{FV02} as the space of joint solutions of the partial differential equations
\begin{equation*}
	L_p\psi = 0,\quad p\in\mathbb{R}\lbrack x_1,\ldots,x_n\rbrack^W.
\end{equation*}
Here, $\mathbb{R}\lbrack x_1,\ldots,x_n\rbrack^W$ denotes the algebra of polynomials invariant under the natural action of the Coxeter group $W$. It is known that all joint solutions of these equations are polynomials, that the dimension of $H_m$ is $|W|$, and that the subspace consisting of the $m$-harmonic polynomials of the highest degree is spanned by the $m$-discriminant
\begin{equation*}
	w_m(x)\equiv\prod_{\alpha\in R_+}(\alpha,x)^{2m+1},
\end{equation*}
see \cite{FV02, EG02}. We note that, for $m$-harmonic polynomials, Theorem \ref{formThm} remains valid if we substitute $p$ for $(e^{L/2}p)$. It follows that
\begin{equation*}
	(w_m,w_m)^{R_+} = \phi(0,0)\int_{\mathbb{R}^n}\prod_{\alpha\in R_+}(\alpha,x)^{2m+2}d\gamma(x).
\end{equation*}
Recalling Corollary \ref{phiCor}, we can thus deduce the value of $(w_m,w_m)^{R_+}$ by combining Theorem~\ref{MMThm} and Proposition \ref{intProp} and making use of \eqref{degprod}. In this way we arrive at the following Proposition.

\begin{proposition}
Let $(x)_n$ be the Pochhammer symbol defined by  $$(x)_n\equiv (x)(x+1)\cdots(x+n-1).$$
Then we have
\begin{equation*}
	(w_m,w_m)^{R_+} = (-1)^{m|R_+|}\prod_{j=1}^n(m+2)_{(m+1)(d_j-1)}(m+1)_{m(d_j-1)}.
\end{equation*}
\end{proposition}

It is clear from \eqref{bilindef} that $(1,1)^{R_+}=1$. This shows that, when restricted to the space of $m$-harmonics $H_m$, the bilinear form $(\cdot,\cdot)^{\mathcal{A}}$ is (at least in general) indefinite. It would thus be interesting to determine its signature.

\subsection{Cases allowing two distinct multiplicities}\label{2pSec}
Here, we shall consider the Coxeter groups corresponding to the root systems $B_n$ and $F_4$, which allow for two distinct multiplicities. (The dihedral root systems are contained as a special case in Section \ref{2DMMSec}.) We recall that the two orbits of these root systems under the action of the Coxeter group consist of 'short' and 'long' roots, respectively. We fix the following choices of positive roots: $R_+\equiv R_{+,s}\cup R_{+,l}$ with short roots $R_{+,s}$ consisting of
\begin{gather*}
	 e_j, 1\leq j\leq n, \quad \mathrm{(B_n)}\\
	e_j, 1\leq j\leq 4,\quad \frac{1}{2}(e_1\pm e_2\pm e_3\pm e_4), \quad \mathrm{(F_4)}\\
\end{gather*}
and long roots $R_{+,l}$ being
\begin{gather*}
	e_j\pm e_k,\ 1\leq j<k\leq n, \quad \mathrm{(B_n)}\\
	e_j\pm e_k,\ 1\leq j<k\leq 4. \, \quad \mathrm{(F_4)}\\
\end{gather*}
To the short roots we associate the multiplicity $m_1$ and to the long $m_2$. The relevant Macdonald-Mehta integrals are stated in the following Theorem.

\begin{theorem}\label{MMThm2}
Let
\begin{equation*}
	\Delta_a(x) = \prod_{\alpha\in R_{a,+}}(\alpha,x),\quad a=s,l.
\end{equation*}
For ${\rm Re}\, k_1,{\rm Re}\, k_2\geq 0$, the integrals
\begin{equation*}
	\int_{\mathbb{R}^n}|\Delta_s(x)|^{2k_1}|\Delta_l(x)|^{2k_2}d\gamma(x)
\end{equation*}
are given by
\begin{gather*}
	2^{-nk_1}\prod_{j=1}^n\frac{\Gamma\big(1+2k_1+2k_2(j-1)\big)}{\Gamma\big(1+k_1+k_2(j-1)\big)}\frac{\Gamma(1+jk_2)}{\Gamma(1+k_2)}, \quad \mathrm{(B_n)}\\
	2^{-12k_1}\frac{\Gamma(4(k_1+k_2)+1)\Gamma(6(k_1+k_2)+1)}{\Gamma(k_1+k_2+1)\Gamma(3(k_1+k_2)+1)}\\ \times\prod_{j=1,2}\frac{\Gamma(2k_j+1)\Gamma(3k_j+1)\Gamma(2k_j+2(k_1+k_2)+1)}{\Gamma(k_j+1)\Gamma(k_j+1)\Gamma(k_j+k_1+k_2+1)}. \quad \mathrm{(F_4)}\\
\end{gather*}
\end{theorem}

\begin{remark}
In the $B_n$ case the evaluation was obtained by Macdonald \cite{Mac82} (see Section 6) as a limiting case of Selberg's integral formula. The $F_4$ case was computed by Garvan \cite{Gar89} using computer calculations. A uniform proof, valid for all crystallographic root systems, was obtained by Opdam \cite{Opd89}. However, the specialisations of his uniform formula to the two cases above takes a somewhat different form, and they are not directly applicable to our discussion below.
\end{remark}

The reflections in the hyperplanes $(\alpha,x)=0$, $\alpha\in R_+$, generate the Weyl group $W$ of the corresponding root system. Given $w\in W$ and a reduced decomposition $w=s_1\cdots s_m$ in terms of simple reflections we let $\ell_1(w)$ and $\ell_2(w)$ denote the number of reflections given by short- and long roots, respectively, see, e.g., Section 2 in Macdonald \cite{Mac72}.

As before, we let $C_+$ denote the 'positive' chamber of $\mathbb{R}^n$, given by the requirement $(\alpha,x)>0$ for all $\alpha\in R_+$. Fixing $\xi\in C_+$, we can thus conclude that the function $\Delta_s(x)^{2k_1}\Delta_l(x)^{2k_2}$ has a (unique) continuous branch on $it\xi+\mathbb{R}^n$, $t>0$, that tends to the positive branch on $C_+$ as $t\downarrow 0$. For this branch, we have the following Lemma.

\begin{lemma}\label{contLemma}
For $\xi\in C_+$ and ${\rm Re}\, k_1, {\rm Re}\, k_2\geq 0$, we have
\begin{equation}\label{anCont2}
	\int_{i\xi+\mathbb{R}^n}\Delta_s(x)^{2k_1}\Delta_l(x)^{2k_2}d\gamma(x) = \frac{P(k_1,k_2)}{|W|}\int_{\mathbb{R}^n}|\Delta_s(x)|^{2k_1}|\Delta_l(x)|^{2k_2}d\gamma(x),
\end{equation}
where $P(k_1,k_2)$ is given by
\begin{gather*}
	\prod_{j=1}^n\frac{1-e^{2\pi i(2k_1+2(j-1)k_2)}}{1-e^{2\pi i(k_1+(j-1)k_2)}}\frac{1-e^{2\pi i jk_2}}{1-e^{2\pi ik_2}}, \quad \mathrm{(B_n)}\\
	\frac{1-e^{2\pi i(4k_1+4k_2)}}{1-e^{2\pi i(k_1+k_2)}}\frac{1-e^{2\pi i(6k_1+6k_2)}}{1-e^{2\pi i(3k_1+3k_2)}}\\ \times\prod_{j=1,2}\frac{1-e^{4\pi ik_j}}{1-e^{2\pi i k_j}}\frac{1-e^{6\pi ik_j}}{1-e^{2\pi i k_j}}\frac{1-e^{2\pi i(2k_j+2k_1+2k_2)}}{1-e^{2\pi i(k_j+k_1+k_2)}}.\quad \mathrm{(F_4)}\\
\end{gather*}
\end{lemma}

\begin{proof}
Proceeding by induction on $\ell_1(w)+\ell_2(w)$, it is straightforward to show that the limit of the relevant branch of $\Delta_s(x)^{2k_1}\Delta_l(x)^{2k_2}$ in the chamber $w(C_+)$, $w\in W$, is given by
\begin{equation*}
	\Delta_s(wx)^{2k_1}\Delta_l(wx)^{2k_2} = \exp\big(2\pi i\lbrack k_1\ell_1(w)+k_2\ell_2(w)\rbrack\big) |\Delta_s(x)|^{2k_1}|\Delta_l(x)|^{2k_2},\quad x\in C_+.
\end{equation*}
Since $W$ acts simply transitively on the chambers of $\mathbb{R}^n$, we can thus conclude that
\begin{equation*}
\begin{split}
	\int_{\mathbb{R}^n}\Delta_s(x)^{2k_1}\Delta_l(x)^{2k_2}d\gamma(x) &= \sum_{w\in W}\int_{C_+}\Delta_s(wx)^{2k_1}\Delta_l(wx)^{2k_2}d\gamma(x)\\ &= \left(\sum_{w\in W}\exp\big(2\pi i\lbrack k_1\ell_1(w)+k_2\ell_2(w)\rbrack\big)\right)\\ & \qquad \times \int_{C_+}|\Delta_s(x)|^{2k_1}|\Delta_l(x)|^{2k_2}d\gamma(x).
\end{split}
\end{equation*}
It follows from Section 2.2 in Macdonald \cite{Mac72} that the sum in the right-hand side is equal to $P(k_1,k_2)$. Finally, let us replace the domain of integration $\mathbb{R}^n$ in the left-hand side by the (real) hyperplane $it\xi+\mathbb{R}^n$, $t>0$. Since the hyperplane does not intersect any of the reflection hyperplanes $(\alpha,x)=0$, $\alpha\in R_+$, the resulting integral is independent of $t$, and the statement follows by taking the limit $t\downarrow 0$.
\end{proof}

Proceeding as in the previous section, making use of the reflection equation $\Gamma(z)\Gamma(1-z)=\pi/\sin\pi z$ and Theorem \ref{MMThm2}, we obtain an expression for the right-hand side of \eqref{anCont2} that it is manifestly analytic for ${\rm Re}\, k_1,{\rm Re}\, k_2<0$. Then setting $k_1=-m_1$ and $k_2=-m_2$ with $m_1,m_2\in\mathbb{N}$ we arrive at the following Proposition. Note that the order of the Weyl group of type $F_4$ is  $2^7\times 3^2.$
\begin{proposition}
Let $\xi\in\mathbb{R}^n$ be such that $(\alpha,\xi)\neq 0$ for all $\alpha\in R_+$, and let $m_1,m_2\in\mathbb{N}$. Then, the integral
\begin{equation*}
	\int_{i\xi+\mathbb{R}^n}\frac{d\gamma(x)}{\Delta_s(x)^{2m_1}\Delta_l(x)^{2m_2}}
\end{equation*}
is given by
\begin{gather*}
	 \frac{(-2)^{nm_1}}{2^nn!}\prod_{j=1}^n\frac{\Gamma\big(m_1+(j-1)m_2\big)}{\Gamma\big(2m_1+2(j-1)m_2\big)}\frac{\Gamma(m_2)}{\Gamma(jm_2)},\quad \mathrm{(B_n)}\\
	\frac{2^{12m_1}}{2^7\times 3^2}\frac{\Gamma(m_1+m_2)\Gamma\big(3(m_1+m_2)\big)}{\Gamma\big(4(m_1+m_2)\big)\Gamma\big(6(m_1+m_2)\big)}\\ \times\prod_{j=1,2}\frac{\Gamma(m_j)^2\Gamma(m_j+m_1+m_2)}{\Gamma(2m_j)\Gamma(3m_j)\Gamma(2m_j+2m_1+2m_2)}.\quad \mathrm{(F_4)}\\
\end{gather*}
\end{proposition}

\section{Macdonald-Mehta integrals -- 2D examples}\label{2DMMSec}

In this section we shall consider configurations of vectors in $\mathbb{R}^2$. Suppose that $\mathcal{A}$ is such a configuration of vectors with multiplicities  for which the Baker-Akhiezer function exists, and consider the corresponding Schr\"odinger type operator
\begin{equation*}
	\mathcal L_{\mathcal A} = -\Delta + \sum_{\alpha\in {\mathcal A}} \frac{m_\alpha(m_\alpha+1)(\alpha,\alpha)}{(\alpha,x)^2}.
\end{equation*}
Introducing polar coordinates $x=(r \cos \varphi, r \sin \varphi)$, we can write this operator in the form
\begin{equation*}
	\mathcal L_{\mathcal A} = -\frac{\partial^2}{\partial r^2} - \frac{1}{r}\frac{\partial}{\partial r} + \frac{1}{r^2}{\mathcal L^{(\varphi)}_{\mathcal A}},\quad{\mathcal  L^{(\varphi)}_{\mathcal A}} = -\frac{\partial^2}{\partial\varphi^2} + V(\varphi)
\end{equation*}
for some potential function $V$. Berest \cite{Ber97} showed that the existence of a corresponding Baker-Akhiezer function implies that the latter operator $\mathcal L^{(\varphi)}_{\mathcal A}$ is obtained from $-\partial^2_\varphi$ by a sequence of Darboux transformations.

Let us fix $m, \tilde m, l, q \in \mathbb{N}$ such that $m\ge \tilde m$, $m, q\ge 1$, and $l$ is even.
Let us consider a sequence of Darboux transformations associated with the functions $\chi_j=\cos (k_j \varphi)$, where $k_j= q j$ for $j=0,\ldots, m-\tilde m$,  $k_{m - \tilde m +j}= q(m - \tilde m +2j)$ for $j=1,\ldots,\tilde m-1$, and $k_{m}=q( m+ \tilde m +l)$. It is known that the corresponding configuration $\mathcal{A}_{(m, \tilde m, 1^l)}^q$ is real and the corresponding lines form a dihedral arrangement of $2q$ lines with multiplicities $m$ and $\tilde m$ together with $q l$ lines of multiplicity 1 that form $l/2$ dihedral orbits, and that the Baker-Akhiezer function exists, see \cite{FJ12}. In particular, the case $l=0$ gives the dihedral configuration with the multiplicities $m, \tilde m$.

We note the following equality of ratios of Wronskians:
\begin{equation*}
	Q\equiv \frac{\mathrm{Wr}[\chi_0,\ldots,\chi_m]}{\mathrm{Wr}[\chi_0,\ldots,\chi_{m-1}]} = -q(m+\tilde m+l) \frac{\mathrm{Wr}[\hat\chi_1,\ldots,\hat\chi_m]}{\mathrm{Wr}[\hat\chi_1,\ldots,\hat\chi_{m-1}]},
\end{equation*}
where $\hat \chi_j = \sin k_j \varphi$. We also have
$$
{\mathrm{Wr}[\hat\chi_1,\ldots,\hat\chi_{m-1}]}=\pm 2^{\frac{(m-1)(m-2)+\tilde m(\tilde m-1)}{2}} (\cos q\vf)^{\frac{\tilde m(\tilde m-1)}{2}} (\sin q\vf)^{\frac{m(m-1)}{2}}  \prod_{\nad{i,j=1}{i>j}}^{m-1} (k_i-k_j),
$$
\begin{multline}
{\mathrm{Wr}[\hat\chi_1,\ldots,\hat\chi_{m}]}=\pm 2^{l+\frac{m(m-1)+\tilde m(\tilde m+1)}{2}} (\cos q\vf)^{\frac{\tilde m(\tilde m+1)}{2}} (\sin q\vf)^{\frac{m(m+1)}{2}} \times \\
 \prod_{j=1}^l \sin(q\vf - \vf_j) \prod_{\nad{i,j=1}{i>j}}^{m} (k_i-k_j),
\end{multline}
so it follows that
\begin{equation}\label{wronski}
	Q = A (\sin q \varphi)^m (\cos q\vf)^{\tilde m} \prod_{j=1}^l \sin (q \varphi - \varphi_j)
\end{equation}
for some angles $\varphi_j\in (0,\pi)$, $\vf_j\ne \pi/2$ such that $\sum_{j=1}^l \varphi_j = \frac{\pi}{2}l$  and with
\begin{equation}
\label{a-coeff}
	A= \pm2^{m+\tilde m+l-1}\prod_{j=0}^{m-1}(k_m-k_j) =\pm q^m 2^{m+\tilde m+l-1} \prod_{j=0}^{m-\tilde m} (m+\tilde m+l-j)\prod_{j=1}^{\tilde m-1} (l+2j).
\end{equation}
 Let  $\varphi_{j,s}=(\varphi_j + \pi s)/q$, where $\vf_0=0$, $\vf_{l+1}=\pi/2$, $j=0,\ldots,l+1$, $s=0,\ldots,q-1$. Then the (normalised)  configuration of vectors $\mathcal{A}_{(m,  \tilde m, 1^l)}^q$ consists of the vectors $\sqrt{2}(-\sin \varphi_{j,s}, \cos \varphi_{j,s})$. The corresponding multiplicities $m_{0,s}=m$, $m_{l+1,s}=\tilde m$ and $m_{j,s}=1$ if $j\neq 0,l+1$.

The Baker-Akhiezer functions corresponding to configurations of vectors in two dimensions were determined by Berest et al.~\cite{BCE08} in terms of Darboux transformations data. It follows from Theorems 2 and 3 in said paper that, for the normalised configuration $\mathcal{A}_{(m, \tilde m,  1^l)}^q$, we have
\begin{equation*}
	\phi(0,0)=(-1)^{q(m+\tilde m+l)}  2^{2q(m+\tilde m+l)-1} c (q(m+\tilde m+l))! \Psi,
\end{equation*}
where
\begin{equation*}
	\Psi=\left(Q^{-1} \prod_{j=0}^{l+1} \prod_{s=0}^{q-1} \sin^{m_{j,s}}(\varphi - \varphi_{j,s})\right)^2
\end{equation*}
and
$$
c=  \prod_{j=0}^{m-1}(k_m^2-k_j^2)  =q^{2m} \prod_{j=0}^{m-\tilde m} ((m+\tilde m+l)^2-j^2)  \prod_{j=1}^{\tilde m-1} ((m+\tilde m+l)^2-(m-\tilde m-2j)^2).
$$

By rearranging \eqref{wronski} into the form
\begin{equation*}
	Q= 2^{(q-1)(l+m+\tilde m)} A  \prod_{j=0}^{l+1} \prod_{s=0}^{q-1} \sin^{m_{j,s}}(\varphi - \varphi_{j,s})
\end{equation*}
we arrive at the following Proposition.

\begin{proposition}\label{2Dphi}
Let $\phi(x,\lambda)$ be the Baker-Akhiezer function corresponding to the configuration  $\mathcal{A}_{(m, \tilde m, 1^l)}^q$. Then it  satisfies
\begin{equation*}
	\phi(0,0)= (-1)^{q(m+\tilde m)} 2 \Gamma(q(m+\tilde m+l)+1) \prod_{j=0}^{m-1}\frac{k_m+k_j}{k_m-k_j}.
\end{equation*}
\end{proposition}

As a direct consequence of Corollary \ref{phiCor}, we obtain an explicit evaluation of the corresponding Macdonald-Mehta type integral.

\begin{theorem}
Fix $m, \tilde m, l, q \in \mathbb{N}$ such that $m\ge \tilde m$; $m, q\ge 1$, and $l$ is even.
Let $\chi_j=\cos (k_j \varphi)$, where   $k_j=q j$ for $j=1,\ldots, m-\tilde m$, $k_{m - \tilde m +j}= q(m - \tilde m +2j)$ for $j=1,\ldots,\tilde m-1$, and $k_{m}=q( m+ \tilde m +l)$. Then the generalised Macdonald-Mehta integral
\begin{equation*}
	M\equiv \frac{1}{2 \pi} \int_{i \xi + \mathbb{R}^2}  \left(\frac{\mathrm{Wr}[\chi_0,\ldots,\chi_{m-1}]}{\mathrm{Wr}[\chi_0,\ldots,\chi_{m}]}\right)^2 \frac{e^{-\frac{x^2}{2}}}{x^{2q(m+\tilde m+l)}} dx
\end{equation*}
is given by
\begin{equation*}
	M=(-1)^{q(m+\tilde m)} \left(2^{q(m+\tilde m +l)-1}  \Gamma(q(m+\tilde m+l)+1) \prod_{j=0}^{m-1} (k_m^2-k_j^2)\right)^{-1},
\end{equation*}
where $\xi \in \mathbb{R}^2$ is an arbitrary vector satisfying $\frac{\mathrm{Wr}[\chi_0,\ldots,\chi_{m}]^2 x^{2q(m+\tilde m+l)}}{\mathrm{Wr}[\chi_0,\ldots,\chi_{m-1}]^2}|_{x=\xi}\ne 0$.
\end{theorem}

\begin{proof}
We have
\begin{equation*}
\begin{split}
	M &= \frac{1}{2 \pi} \int_{i \xi + \mathbb{R}^2}  \frac{e^{-\frac{x^2}{2}}}{Q^2 x^{2q(m+\tilde m+l)}}dx\\ &= \frac{1}{2 \pi} \int_{i \xi + \mathbb{R}^2}  \frac{e^{-\frac{x^2}{2}} dx}{A^2 2^{(q-2)(m+\tilde m+l)}\prod_{\alpha_{j,s}\in {\mathcal A}_{(m, \tilde m, 1^l)}^q} (\alpha_{j,s},x)^{2m_{j,s}}},
\end{split}
\end{equation*}
where $Q$ and $A$ are given by formulas (\ref{wronski}), (\ref{a-coeff}). By Corollary \ref{phiCor}, we have
$$
M= (A^2 2^{(q-2)(m+\tilde m+l)} \phi(0,0))^{-1},
$$
where $\phi(0,0)$ is found in Proposition \ref{2Dphi}.
\end{proof}

\section{Deformed root systems and Dotsenko-Fateev integral}\label{SectDeform}

In dimensions higher than 2 there are two known series of non-Coxeter configurations that admit the Baker-Akhiezer functions \cite{CFV96}, \cite{CFV98}, \cite{CFV99}.

The first series ${\mathcal A}_m(p)\subset{\mathbb R}^{m+1}$ depends on one parameter $p\in \mathbb N$ and it consists of the vectors $e_i-e_j$, $1\le i<j \le m$ with the multiplicites $p$ and the vectors $e_i-\sqrt{p}e_{m+1}$, $1\le i \le m$ with the multiplicities 1. It can be viewed as a special deformation of the ${A}_m$ root system.

 The second series ${\mathcal C}_{m+1}(r,s)\subset {\mathbb R}^{m+1}$ depends on two parameters $r,s \in \mathbb N$ such that $p=\frac{2r+1}{2s+1}\in \mathbb N$. It can be viewed as a special deformation of the ${C}_{m+1}$ (or ${B}_{m+1}$) root system. It consists of the vectors $e_i\pm e_j$, $1\le i <j \le m$ with the multiplicities $p$, the vectors $e_i$, $1\le i \le m$ with the multiplicities $r$, the vector $e_{m+1}$ with the multiplicity $s$ and the vectors $e_i\pm \sqrt{p}e_{m+1}$, $1\le i \le m$ with the multiplicities 1.

These arrangements correspond to the {\it deformed root systems} from \cite{SV} with integer multiplicities.
Recall that according to \cite{SV} one can define for any basic classical Lie superalgebra a certain deformation of the corresponding root system by prescribing multiplicities to the roots and changing the bilinear form. The corresponding deformed root systems form two series of type $A(n,m)$ and $BC(n,m)$ and 3 exceptional cases. All the multiplicities are integer only in the cases listed above.

The main property of the deformed root systems, which will be important for us, is that the corresponding deformed Calogero-Moser operator
\begin{equation}
\label{hCM}
L_{\mathcal A}=-\Delta+\sum_{\alpha \in \mathcal A}\frac{m_\alpha(m_\alpha+1)(\alpha,\alpha)}{(\alpha,x)^2}+x^2
\end{equation}
has a pseudo-groundstate of the form
\begin{equation}
\label{pgCM}
\psi_0=\prod_{\alpha \in \mathcal A}(\alpha,x)^{-m_\alpha}e^{-\frac{1}{2}x^2}
\end{equation}
 (cf. property 3 in the original trigonometric version in \cite{SV}).
 Indeed, one can check that the main identity (12) from \cite{SV} implies that
 $$L_{\mathcal A}\psi_0= \lambda_0 \psi_0, \quad \lambda_0=N-2 \sum_{\alpha \in \mathcal A}m_\alpha,$$
 where $N$ is the dimension of the space.
 Since the multiplicity of the isotropic roots are fixed to be 1, such a function always has a singularity at the corresponding hyperplanes. Nevertheless, sometime one can make sense of its norm as the integral
\begin{equation}
\label{norm}
I_{\mathcal A}=|\psi_0|^2=\int_{i\xi+\mathbb{R}^n} dx \prod_{\alpha \in \mathcal A}(\alpha,x)^{-2m_\alpha}e^{-x^2},
\end{equation}
in the same way as we did before. In particular for the deformed system of type $A(n,m)$ we have
 \begin{equation}
\label{defamn}
I_{A(n,m)} =  \int_{{\mathbb R}^n+i\eta} dx \int_{{\mathbb R}^m+i \xi} dy \frac{\prod_{i<j}^n (x_i - x_j)^{-2k} \prod_{i<j}^m (y_i- y_j)^{-\frac{2}{k}} e^{-\sum_{j=1}^n x_j^2-\frac{1}{k}\sum_{i=1}^m y_i^2}}{\prod_{i=1}^n \prod_{j=1}^m (x_i - y_j)^2},
\end{equation}
 where $k$ is a parameter of deformation. When $k=-1$ we have the root system of Lie superalgebra $sl(n,m),$ while $k=1$ corresponds to the usual root system $A_{m+n-1}$.

 In $BC(n,m)$ case we have the integral
\begin{equation}\label{bcdef}
I_{BC(n,m)} =\int_{{\mathbb R}^n+i \xi} dx \int_{{\mathbb R}^m+i\eta} dy  \frac{  \prod_{i=1}^n x_i^{-2r} \prod_{j=1}^m y_j^{-2s} e^{-\sum_{i=1}^n x_i^2-\frac{1}{k}\sum_{j=1}^m y_j^2}}{\prod_{i<j}^n (x_i^2 - x_j^2)^{2k}  \prod_{i<j}^m (y_i^2- y_j^2)^{\frac{2}{k}}   \prod_{j=1}^m \prod_{i=1}^n ( x_i^2 -y_j^2)^2},
\end{equation}
where the parameters $k,r,s$ satisfy one relation
$$2r+1=k(2s+1)$$
 (see \cite{SV} and take into account that proportional roots can be combined in the rational case).

We are going to show now that these integrals can be computed explicitly using the following remarkable formula found by Dotsenko and Fateev \cite{DF} in conformal field theory.

For $n,m \in \mathbb N, \,\, \alpha, \rho \in \mathbb C$ consider the following integral, which can be interpreted as a norm of the pseudo-ground state $\psi_0$ of the deformed trigonometric $BC(n,m)$ Calogero-Moser system \cite{SV} (cf. \cite{Mac82}, where Selberg integral was interpreted in a similar way):
\begin{equation}
\label{DF-intergral}
J=\prod_{i=1}^n \int_{C_i} d t_i \prod_{j=1}^m \int_{S_j} d \tau_j \frac{\prod_{i=1}^n t_i^{-\frac{\alpha}{\rho}}(1-t_i)^{-\frac{\beta}{\rho}} \prod_{j=1}^m \tau_j^{\alpha}(1-\tau_j)^{\beta} \prod_{i<j}^n (t_i-t_j)^\frac{2}{\rho}  \prod_{i<j}^m (\tau_i-\tau_j)^{2\rho}}{\prod_{j=1}^m \prod_{i=1}^n (t_i-\tau_j)^2},
\end{equation}
where the contours $C_i$ are arcs going from the point 0 to the point 1 in the upper half plane such that the arc $C_i$ is above the arc $C_{i-1}$. Similarly, the contours $S_j$ are arcs going from the point  0 to the point 1 in the lower half-plane such that the arc $S_j$ is above the arc $S_{j-1}$, see Fig. 1.

\begin{figure}[h]
\centerline{ \includegraphics[width=12cm]{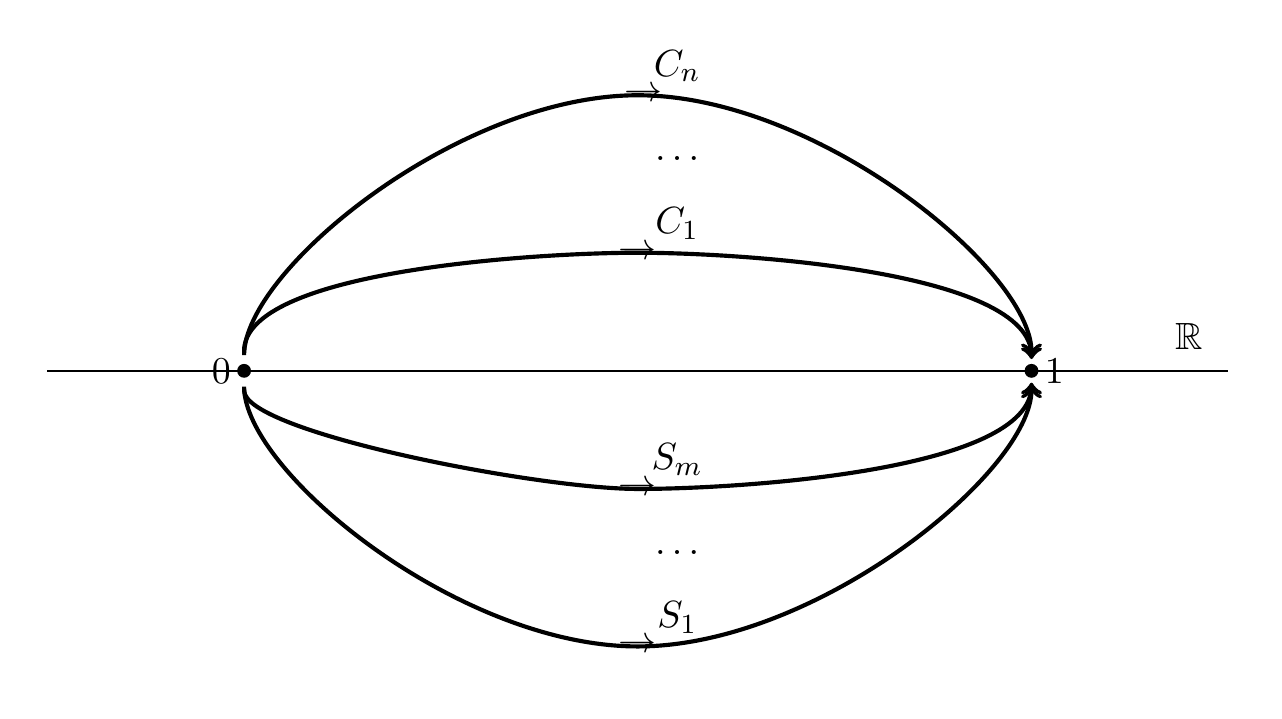} }
\caption{Integration contour for Dotsenko-Fateev integral} \label{DFcontour}
\end{figure}

Note that the integrand of $J$ is single-valued on the contours so the natural branch is chosen that is $x^\rho>0$ if $x>0$, $\rho\in\mathbb R$, etc.  Suppose that $\rho \in \mathbb{R}$ and 
\begin{equation}
\label{domain}
 \rho <0, \quad {\rm Re}\, \alpha>n-1-(m-1) \rho, \quad {\rm Re}\, \beta>n-1-(m-1) \rho.
\end{equation}
Then all the powers in the integrand are well defined and the integrals converge (cf. \cite{DF}). Dotsenko and Fateev  \cite[Appendix A]{DF} showed that
\begin{multline}\label{DFth}
J=\rho^{2nm} \prod_{k=1}^m \frac{1-e^{-2\pi i k \rho}}{1-e^{-2\pi i \rho}} \prod_{k=1}^n \frac{1-e^{-2\pi i k /\rho}}{1-e^{-2\pi i/ \rho}} \prod_{j=1}^n\frac{\Gamma(j/\rho)}{\Gamma(1/\rho)} \prod_{j=1}^m \frac{\Gamma(j \rho -n)}{\Gamma(\rho)} \\ \times
\prod_{j=0}^{n-1}\frac{\Gamma(1-\alpha/\rho+j/\rho) \Gamma(1-\beta/\rho+j/\rho)}{\Gamma(2-2m-(\alpha+\beta+n-1+j)/\rho)}  \prod_{j=0}^{m-1}\frac{\Gamma(1-n+\alpha+j \rho) \Gamma(1-n+\beta+j \rho)}{\Gamma(2-n+2\alpha+(m-1+j)\rho)}.
\end{multline}

\subsection{Deformations of type $A$}
Following Bombieri's idea we are going to take a special limit of the Dotsenko-Fateev integral that leads to a generalised Macdonald-Mehta integral for the configurations ${\mathcal A}_p(q).$
Assume that $\alpha=\beta \in {\mathbb R}_{>0}$, and $\alpha$ is large, so we are in the domain \eqref{domain}.

 Firstly we change the variables in the integral $J$. Define $\hat t_i = (2 t_i-1)L$, $\hat \tau_j =(2 \tau_j-1)L$, where $L\in {\mathbb R}_{>0}$. Let $\hat C_i, \hat S_j$ be the corresponding rescaled contours that now connect the points $-L$ and $L$. Then
$$
J=a\prod_{i=1}^n \int_{\hat C_i} d \hat t_i \prod_{j=1}^m \int_{\hat S_j} d \hat \tau_j \frac{\prod_{i=1}^n (1-\frac{\hat t_i^2}{L^2})^{-\frac{\alpha}{\rho}} \prod_{j=1}^m (1-\frac{\hat \tau_j^2}{L^2})^{\alpha} \prod_{i<j}^n (\hat t_i - \hat t_j)^\frac{2}{\rho}  \prod_{i<j}^m (\hat\tau_i- \hat\tau_j)^{2\rho}}{\prod_{j=1}^m \prod_{i=1}^n (\hat t_i - \hat \tau_j)^2},
$$
where
$$
a= L^{-\rho m (m-1)-\frac{1}{\rho}n(n-1)+2 m n-m-n} 2^{-\rho m (m-1)-\frac{1}{\rho}n(n-1)+2 m n-m-n+\frac{2\alpha}{\rho}n-2\alpha m}.
$$
Recall that $\lim_{L\to\infty}(1+\frac{x}{L})^L=e^x$. This allows us to take the limit of $J$ to get the generalised Macdonald-Mehta integral.

\begin{lemma}\label{lem1}
Let $\alpha=L^2/2$ where $L\in {\mathbb R}_+$. Then
$$
\lim_{L\to \infty} a^{-1}J= \int_{{\mathbb R}^n+i \xi}  d t \int_{{\mathbb R}^m+i\eta}    d \tau \frac{\prod_{i<j}^n (t_i - t_j)^\frac{2}{\rho}  \prod_{i<j}^m (\tau_i- \tau_j)^{2\rho} e^{-\frac12\sum_{j=1}^m\tau_j^2+\frac{1}{2\rho}\sum_{i=1}^n t_i^2}}{\prod_{j=1}^m \prod_{i=1}^n (t_i -  \tau_j)^2},
$$
where $\xi\in  {\mathbb R}^n$, $\eta \in {\mathbb R}^m$ are  such that $\xi_n>\ldots>\xi_1>\eta_m>\ldots>\eta_1$.
\end{lemma}

Now we analyze the limit of the right-hand side of \eqref{DFth}.

\begin{lemma}\label{lem2}
Let $\alpha=\frac{L^2}{2}$. Then
$$
\lim_{L\to \infty} b \prod_{i=0}^{n-1}\frac{\Gamma(1-\alpha/\rho+i/\rho)^2}{\Gamma(2-2m-2\alpha/\rho+\frac{n-1+i}{\rho})}  \prod_{j=0}^{m-1}\frac{\Gamma(1-n+\alpha+j \rho)^2}{\Gamma(2-n+2\alpha+(m-1+j)\rho)}=1,
$$
where $b=\sqrt\frac{\pi}{2}^{-m-n}2^{-2n m +m L^2+m(m-1)\rho-\frac{L^2 n}{\rho}+\frac{n(n-1)}{\rho}} L^{m+2n m +m(m-1)\rho} (-\frac{L^2}{\rho})^{-\frac{n}{2}-2 m n +\frac{n(n-1)}{2\rho}}$.
\end{lemma}

\begin{proof}
The Stirling formula $\Gamma(x+1)\sim \sqrt{2\pi x}(\frac{x}{e})^x, x\to +\infty$, leads to the relations
$$
\frac{\Gamma(1-\alpha/\rho+i/\rho)^2}{\Gamma(2-2m-2\alpha/\rho+\frac{n-1+i}{\rho})}\sim \sqrt\frac{\pi}{2} 2^{\frac{L^2-2i}{\rho}} (-\frac{L^2}{\rho})^{-\frac12+2m-\frac{n-1-i}{\rho}},
$$
$$
\frac{\Gamma(1-n+\alpha+j \rho)^2}{\Gamma(2-n+2\alpha+(m-1+j)\rho)}\sim \sqrt\frac{\pi}{2} 2^{2n-L^2-2j \rho} L^{-1-2n-2(m-1-j)\rho},
$$
and the lemma follows.
\end{proof}
The previous results allow to calculate the following generalised Macdonald-Mehta integral, which is a version of the norm integral (\ref{defamn}), corresponding to the deformed root system of type $A(n,m)$:
\begin{equation}
\label{gMM}
M= \frac{1}{{(2\pi)}^\frac{m+n}{2}} \int_{{\mathbb R}^n+i \xi} dt \int_{{\mathbb R}^m+i\eta} d\tau \frac{\prod_{i<j}^n (t_i - t_j)^\frac{2}{\rho}  \prod_{i<j}^m (\tau_i- \tau_j)^{2\rho} e^{-\frac12\sum_{j=1}^m\tau_j^2-\frac{1}{2}\sum_{i=1}^n t_i^2}}{\prod_{j=1}^m \prod_{i=1}^n (\sqrt{-\rho} t_i -  \tau_j)^2},
\end{equation}
where $\xi \in  {\mathbb R}^n$, $\eta \in  {\mathbb R}^m$ are such that $\xi_n>\ldots>\xi_1>\eta_m>\ldots>\eta_1$.

\begin{proposition}\label{MDF}
Let $M$ be given by \eqref{gMM}. Then
$$
M= \varepsilon\prod_{i=1}^n\frac{\Gamma(1-1/\rho)}{\Gamma(1-i/\rho)} \prod_{j=1}^m \frac{\Gamma(1- \rho)}{\Gamma(n+1-j\rho)} = \varepsilon\prod_{j=1}^m \prod_{i=1}^n \frac{1}{i-j \rho} \prod_{i=1}^n\frac{\Gamma(1-1/\rho)}{\Gamma(1-i/\rho)} \prod_{j=1}^m \frac{\Gamma(1- \rho)}{\Gamma(1-j\rho)},
$$
where $\varepsilon = (-1)^m e^{-\pi i(m(m-1)\rho+\frac{n(n-1)}{2\rho})}$.
\end{proposition}

\begin{proof}
In the notations of Lemmas \ref{lem1}, \ref{lem2} we have
\begin{multline}\label{mmm1}
\lim_{L\to\infty} a^{-1}J=\lim_{L\to\infty} b (2\pi)^\frac{m+n}{2} (-\rho)^{\frac{n}{2}-2 m n +\frac{n(n-1)}{2\rho}} J=\\
(2\pi)^{\frac{m+n}{2}} (-\rho)^{\frac{n}{2}+\frac{n(n-1)}{2\rho}} \prod_{k=1}^m \frac{1-e^{-2\pi i k \rho}}{1-e^{-2\pi i \rho}} \prod_{k=1}^n \frac{1-e^{-2\pi i k /\rho}}{1-e^{-2\pi i/ \rho}} \prod_{j=1}^n\frac{\Gamma(j/\rho)}{\Gamma(1/\rho)} \prod_{j=1}^m \frac{\Gamma(j \rho -n)}{\Gamma(\rho)}
\end{multline}
by \eqref{DFth} and by Lemma \ref{lem2}. On the other hand, by Lemma \ref{lem1} and by performing the change of variables $t_i \to \sqrt{-\rho} t_i$ we have
\begin{equation}
\label{mmm2}
\lim_{L\to\infty} \frac{a^{-1}J}{(2\pi)^{\frac{m+n}{2}}}=(-\rho)^{\frac{n}{2}+\frac{n(n-1)}{2\rho}} M.
\end{equation}
It follows from \eqref{mmm1} and \eqref{mmm2} that
$$
M=\prod_{k=1}^m \frac{1-e^{-2\pi i k \rho}}{1-e^{-2\pi i \rho}} \prod_{k=1}^n \frac{1-e^{-2\pi i k /\rho}}{1-e^{-2\pi i/ \rho}} \prod_{i=1}^n\frac{\Gamma(i/\rho)}{\Gamma(1/\rho)} \prod_{j=1}^m \frac{\Gamma(j \rho -n)}{\Gamma(\rho)}.
$$
The statement follows by applying the reflection equation $\Gamma(z)\Gamma(1-z)=\pi/\sin\pi z$.
\end{proof}

\begin{remark}
Due to the quasi-invariance of the integrand in \eqref{gMM} at $\sqrt{-\rho} t_i=\tau_j$ the Proposition \ref{MDF} remains valid if the parameters $\xi, \eta$ satisfy the weaker restrictions  $\xi_n>\ldots>\xi_1$, $\eta_m>\ldots>\eta_1$, $\xi_i\ne \eta_j  \, \forall i,j$.
\end{remark}

Suppose now that $\rho=-p$ where $p\in \mathbb N$. Suppose also that $n=1$.  In this case the integral \eqref{gMM} has the same value if parameters $\xi, \eta$ satisfy the weaker restrictions $\eta_i \ne \eta_j, \xi_1 \ne \eta_i \, \forall i,j$.
Furthermore,
the integral $M$ is of the form that appears in Corollary \ref{phiCor} for the configuration ${\mathcal A}_m(p)$ which admits the Baker-Akhiezer function  \cite{CFV96}.
Thus
 Proposition \ref{MDF} implies the following statement.
 \begin{corollary}
Let $\phi$ be the Baker-Akhiezer function for the configuration ${\mathcal A}_m(p)$. Then
$$
\phi(0,0)= (-1)^{m+\frac{p m (m-1)}{2}} \prod_{j=1}^m \frac{\Gamma(pj+2)}{\Gamma(p+1)}.
$$
\end{corollary}

\subsection{Deformations of type $B$.}
Firstly we change the variables in the integral \eqref{DF-intergral}. Let $\hat t_i= L t_i$, $\hat\tau_i= L \tau_i$, $\beta = L/2$, where $L\in {\mathbb R}_+$. We suppose that $\beta$ and $\alpha \in \mathbb R$  are sufficiently large, so that the conditions \eqref{domain} are satisfied.
Then
\begin{equation}
\label{M1}
\lim_{L \to \infty} L^{m\alpha-\frac{n \alpha}{\rho}+\rho m(m-1)+\frac{n(n-1)}{\rho}-2m n+m+n} J = M_1,
\end{equation}
where $M_1$ is the integral
\begin{equation}\label{m1m1}
  M_1 = \prod_{i=1}^n\int_{\hat C_i} d \hat t_i \prod_{j=1}^m\int_{\hat S_j}  d \hat \tau_j \frac{\prod_{i=1}^n \hat t_i^{-\frac{\alpha}{\rho}} \prod_{j=1}^m \hat \tau_j^{\alpha}e^{-\frac12\sum_{j=1}^m\hat\tau_j+\frac{1}{2\rho}\sum_{i=1}^n \hat t_i}}{ \prod_{i<j}^n (\hat t_i - \hat t_j)^{-\frac{2}{\rho}}  \prod_{i<j}^m (\hat\tau_i- \hat\tau_j)^{-2\rho} \prod_{j=1}^m \prod_{i=1}^n (\hat t_i -  \hat \tau_j)^2},
\end{equation}
for the corresponding scaled contours $\hat C_j$ and $\hat S_j$ connecting $0$ and $\infty$.

\begin{lemma}\label{M1-calculate}
The integral $M_1$ has the value
\begin{multline}\label{M1-value}
M_1=\rho^{2 n m} (-2\rho)^{n(1-2m-\frac{\alpha-n+1}{\rho})} 2^{m(1+\alpha+(m-1)\rho)}
\prod_{k=1}^m \frac{1-e^{-2\pi i k \rho}}{1-e^{-2\pi i \rho}} \prod_{k=1}^n \frac{1-e^{-2\pi i k /\rho}}{1-e^{-2\pi i/ \rho}} \\
\times \prod_{j=1}^n\frac{\Gamma(j/\rho)}{\Gamma(1/\rho)} \prod_{j=1}^m \frac{\Gamma(j \rho -n)}{\Gamma(\rho)} \prod_{j=0}^{n-1} \Gamma(1-\frac{\alpha-j}{\rho})  \prod_{j=0}^{m-1} \Gamma(1-n+\alpha+j\rho).
\end{multline}
\end{lemma}

\begin{proof}
We apply Stirling's formula to the terms in the left-hand side of \eqref{M1}. Since $\lim_{L\to\infty}(1+x/L)^L=e^x$ we get
$$
\frac{\Gamma(1-\frac{\beta-j}{\rho})}{\Gamma(2-2m-(\alpha+\beta-n+1-j)/\rho)}
\sim \left(-\frac{L}{2\rho}\right)^{-1+2m+(\alpha-n+1)/\rho},
$$
$$
\frac{\Gamma(1-n+\beta+j\rho)}{\Gamma(2-n+\alpha+\beta+(m-1+j)\rho)}
\sim \left(\frac{L}{2}\right)^{-1-\alpha-(m-1)\rho},
$$
as $\beta=L/2\to\infty$. The statement follows.
\end{proof}
Now we change variables in \eqref{m1m1} according to $\hat \tau_i = \widetilde\tau_i^2$, $ \hat t_j = \widetilde t_j^2$. Removing an overall factor $2^{m+n}$ from the resulting integral, we obtain
\begin{equation}\label{m2m2}
  M_2 = \prod_{i=1}^n\int_{\hat C_i} d \widetilde t_i \prod_{j=1}^m\int_{\hat S_j} d \widetilde \tau_j \frac{\prod_{i=1}^n \widetilde t_i^{-2\frac{\alpha}{\rho}+1} \prod_{j=1}^m \widetilde \tau_j^{2\alpha+1} e^{-\frac12\sum_{j=1}^m\widetilde\tau_j^2+\frac{1}{2\rho}\sum_{i=1}^n \widetilde t_i^2}}{\prod_{i<j}^n ({\widetilde t}_i^2 - \widetilde t_j^2)^{-\frac{2}{\rho}}  \prod_{i<j}^m (\widetilde\tau_i^2- \widetilde\tau_j^2)^{-2\rho} \prod_{j=1}^m \prod_{i=1}^n (\widetilde t_i^2 -  \widetilde \tau_j^2)^2},
\end{equation}
where the integration contours are chosen as in  \eqref{m1m1}. In particular, this means that the value of the integral $M_2$ is given by the right-hand side of \eqref{M1-value} multiplied by $2^{m+n}$.

Consider now the following version of Macdonald-Mehta integral (\ref{bcdef}),
corresponding to the deformed root system of type $BC(n,m).$

\begin{proposition}\label{main-prop-B}
Let $M$ be the generalised Macdonald-Mehta integral defined by
\begin{multline}\label{Mdef}
M=(2\pi)^{-\frac{m+n}2}\\ \times \int_{{\mathbb R}^n+i \xi} dt \int_{{\mathbb R}^m+i\eta} d\tau  \frac{  \prod_{i=1}^n t_i^{-2\frac{\alpha}{\rho}+1} \prod_{j=1}^m \tau_j^{2\alpha+1} e^{-\frac12\sum_{j=1}^m\tau_j^2-\frac{1}{2}\sum_{i=1}^n t_i^2}}{\prod_{i<j}^n (t_i^2 - t_j^2)^{-\frac{2}{\rho}}  \prod_{i<j}^m (\tau_i^2- \tau_j^2)^{-2\rho}   \prod_{j=1}^m \prod_{i=1}^n (\rho t_i^2 +  \tau_j^2)^2},
\end{multline}
where $\rho<0$, $\alpha \in \mathbb R$, and  $\xi_n>\ldots>\xi_1>\eta_m>\ldots>\eta_1>0$.
Then
\begin{multline}\label{MdefB}
M=(2\pi)^{\frac{m+n}{2}} 2^{2(m+n)-2m n -\frac{n\alpha}{\rho}+\frac{n(n-1)}{\rho}+m\alpha+m(m-1)\rho} e^{\pi i (\frac{m+n}{2}+m\alpha-\frac{n\alpha}{\rho})}\\ \times \prod_{k=1}^m \frac{\Gamma(1-\rho)}{\Gamma(1-k \rho) \prod_{j=1}^n (k\rho-j)}
\prod_{k=1}^n \frac{\Gamma(1-\rho^{-1})}{\Gamma(1-k \rho^{-1})}\\ \times \prod_{j=0}^{n-1}\frac{1}{\Gamma(\frac{\alpha-j}{\rho})} \prod_{j=0}^{m-1} \frac{1}{\Gamma(-\alpha-j \rho)} \prod_{j=0}^{m-1}\prod_{k=0}^{n-1} \frac{1}{\alpha+j \rho -k}.
\end{multline}
\end{proposition}

\begin{proof}
Initially, we assume that $\alpha>n-1-(m-1)\rho$, so that the previous analysis as well as Lemma~\ref{M1-calculate} can be applied.

Rescaling the variables $t_j$, $j=1,\ldots n$, in the integral \eqref{Mdef}, we obtain
\begin{multline*}
	M=(2\pi)^{-\frac{m+n}{2}}(-\rho)^{\frac{\alpha n}{\rho}-\frac{n(n-1)}{\rho}-n}\\ \times \int_{{\mathbb R}^n+i \xi} dt \int_{{\mathbb R}^m+i\eta} d\tau \frac{\prod_{i=1}^n  t_i^{-2\frac{\alpha}{\rho}+1} \prod_{j=1}^m \tau_j^{2\alpha+1} e^{-\frac12\sum_{j=1}^m\tau_j^2+\frac{1}{2\rho}\sum_{i=1}^n t_i^2}}{\prod_{i<j}^n (t_i^2 -  t_j^2)^{-\frac{2}{\rho}}  \prod_{i<j}^m (\tau_i^2- \tau_j^2)^{-2\rho} \prod_{j=1}^m \prod_{i=1}^n (t_i^2 - \tau_j^2)^2}
\end{multline*}
Let $I_{n,m}(t,\tau)$ denote the integrand in the right-hand side, so that
\begin{equation}
\label{Minteg}
	M = (2\pi)^{-\frac{m+n}{2}}(-\rho)^{\frac{\alpha n}{\rho}-\frac{n(n-1)}{\rho}-n}\int_{{\mathbb R}^n+i \xi}dt \int_{{\mathbb R}^m+i\eta} d\tau I_{n,m}(t,\tau).
\end{equation}
Let $S_j^-$ (respectively, $C_j^-$) be the contour $\hat S_j$ (respectively, $\hat C_j$) reflected about the imaginary axis and oriented from left to right. Let also $S_j^+$ (respectively, $C_j^+$) be the contour $\hat S_j$ (respectively, $\hat C_j$) reflected about the real axis and oriented from left to right. By deforming and splitting each integration contour in \eqref{Minteg} into a 'negative' and a 'positive' part, we can rewrite the integral as
\begin{equation*}
	\prod_{j=1}^n\left(\int_{C_j^-}dt_j + \int_{\hat C_j}dt_j\right)\prod_{k=1}^m\left(\int_{S_k^-}d\tau_k + \int_{\hat S_k}d\tau_k\right)I_{n,m}(t,\tau).
\end{equation*}
Expanding the products, we obtain the double-sum
\begin{equation*}
\begin{split}
	\sum_{J\subset\{1,\ldots,n\}}\sum_{K\subset\{1,\ldots,m\}} & \prod_{j\in J}\int_{C_j^-}dt_j\prod_{j^\prime \in \{1,\ldots,n\}\setminus J}\int_{\hat C_{j^\prime}}dt_{j^\prime}\\ & \times\prod_{k\in K}\int_{S_k^-}d\tau_k\prod_{k^\prime \in \{1,\ldots,m\}\setminus K}\int_{\hat S_{k^\prime}}d\tau_{k^\prime}I_{n,m}(t,\tau).
\end{split}
\end{equation*}
(Empty products are identified with $1$.)

Let us fix two such index sets $J$ and $K$, and consider the corresponding summand. Clearly, we may deform the contours ${\mathbb R}_{\geq 0}+i\xi_{j^\prime}$, $j^\prime\notin J$, and ${\mathbb R}_{\geq 0}+i\eta_{k^\prime}$, $k^\prime\notin K$, into $\hat{C}_{j^\prime}$ and $\hat{S}_{k^\prime}$, respectively, without altering the value of the integral. In order to relate $M$ to $M_2$, and thereby be able to make use of Lemma \ref{M1-calculate}, we apply change of variables $t_j\to -t_j$, $j\in J$, and $\tau_k\to -\tau_k$, $k\in K$, and we replace contours $C_j^-$, $S_k^-$ with the contours $C_j^+$, $S_k^+$  respectively. To this end, we note that
\begin{multline*}
	\prod_{j\in J}\int_{C_j^-}dt_j\prod_{k\in K}\int_{S_k^-}d\tau_k I_{n,m}(t,\tau)\\ = \prod_{j\in J}\big(-e^{2\pi i(-\alpha/\rho)}\big)\int_{C_j^+}dt_j\prod_{k\in K}\big(-e^{2\pi i\alpha}\big)\int_{S_k^+}d\tau_k I_{n,m}(t,\tau),
\end{multline*}
where the exponential factors in the right-hand side arise from the factors $t_j^{-2\alpha/\rho+1}$ and $\tau_k^{2\alpha+1}$ in $I_{n,m}(t,\tau)$. We proceed to first deform the contours $C_j^+$ into $\hat C_j$ and then to deform the contours  $S_k^+$ into $\hat S_k$. In doing so we will reverse the order of some contours. Thus the contours $C_j^+$ will have to be swaped with the contours $\hat C_1, \ldots, \hat C_{j-1}$ where $j\in J$ and we start with the smallest $j$. Since a swap of $C$ and $S$ contours does not affect the integral, the product $\prod_{i<j}(t_i-t_j)^{2/\rho}$ in $I_{n,m}(t,\tau)$ will yield a factor $\prod_{j\in J}e^{2\pi i(j-1)/\rho}$.
Similarly, the product $\prod_{i<j}(\tau_i-\tau_j)^{2\rho}$ yields $\prod_{k\in K}e^{2\pi i(k-1)\rho}$. We can thus conclude that
\begin{equation*}
	M=(2\pi)^{-\frac{m+n}{2}}(-\rho)^{\frac{\alpha n}{\rho}-\frac{n(n-1)}{\rho}-n}BM_2
\end{equation*}
with
\begin{equation*}
\begin{split}
	B &= \sum_{J\subset\{1,\ldots,n\}}\prod_{j\in J}\big(-e^{2\pi i((j-1)-\alpha)/\rho}\big)\sum_{K\subset\{1,\ldots,m\}}\prod_{k\in K}\big(-e^{2\pi i((k-1)\rho+\alpha)}\big)\\ & = \prod_{j=0}^{n-1}\big(1-e^{2\pi i(j-\alpha)/\rho}\big)\prod_{k=0}^{m-1}\big(1-e^{2\pi i(k\rho+\alpha)}\big)\\& = 2^{m+n}\epsilon\prod_{j=0}^{n-1}\sin\pi\big(j-\alpha)/\rho\prod_{k=0}^{m-1}\sin\pi(k\rho+\alpha),
\end{split}	
\end{equation*}
where $\epsilon=e^{\pi i (-\frac{m+n}{2}+m\alpha+\frac{m(m-1)}{2}\rho-\frac{n\alpha}{\rho}+\frac{n(n-1)}{2\rho})}$.

We have
$$
\Gamma(1-\frac{\alpha-j}{\rho})\sin\pi(\frac{-\alpha+j}{\rho})=-\frac{\pi}{\Gamma(\frac{\alpha-j}{\rho})},
$$
and
$$
\Gamma(1-n+\alpha+j\rho)\sin\pi(\alpha+j\rho)=\frac{-\pi}{\Gamma(-\alpha-j\rho)\prod_{k=0}^{n-1} (\alpha+j \rho-k)}.
$$
With the use of Lemma \ref{M1-calculate} we thus deduce
\begin{multline*}
M=(2\pi)^{-\frac{m+n}{2}} (-\rho)^{\frac{\alpha n}{\rho}-\frac{n(n-1)}{\rho}-n} 2^{m+n} \rho^{2mn}(-2\rho)^{n(1-2m-\frac{\alpha-n+1}{\rho})} 2^{m(1+\alpha+(m-1)\rho)}\times\\
\prod_{k=1}^m \frac{1-e^{-2\pi i k \rho}}{1-e^{-2\pi i \rho}} \prod_{k=1}^n \frac{1-e^{-2\pi i k /\rho}}{1-e^{-2\pi i/ \rho}}
\prod_{j=1}^n\frac{\Gamma(j/\rho)}{\Gamma(1/\rho)} \prod_{j=1}^m \frac{\Gamma(j \rho -n)}{\Gamma(\rho)} \times \\
(-2\pi)^{m+n}\epsilon
 \prod_{j=0}^{n-1}\frac{1}{\Gamma(\frac{\alpha-j}{\rho})} \prod_{j=0}^{m-1} \frac{1}{\Gamma(-\alpha-j \rho) \prod_{k=0}^{n-1} (\alpha+j \rho -k)},
\end{multline*}
and the statement follows for large values of $\alpha$. Since the integral $M$ and its value \eqref{MdefB} are analytic functions of $\alpha$ the statement is true for any $\alpha$.
\end{proof}

Suppose now that $\rho=-p$ where $p\in \mathbb N$. Suppose also that $n=1$. Let also $\alpha=-r-\frac12$ where $r\in \mathbb N$ and suppose that $s=\frac{\alpha}{\rho}-\frac12 \in \mathbb N$.  In this case the integral $M$ has the same value \eqref{MdefB}
 if parameters $\xi_i, \eta_j$ are non-zero and satisfy the weaker restrictions $\eta_i \ne \eta_j, \xi_1 \ne \eta_i$, for all $i,j=1,\ldots,m$ such that $i\neq j$.
Furthermore,
the integral $M$ is of the form that appears in Corollary \ref{phiCor} for the configuration ${\mathcal C}_{m+1}(r,s)$ which admits the Baker-Akhiezer function  \cite{CFV99}.
Thus
 Proposition \ref{main-prop-B} implies the following statement.
 \begin{corollary}
Let $\phi$ be the Baker-Akhiezer function for the configuration ${\mathcal C}_{m+1}(r,s)$.  Then
\begin{multline*}
\phi(0,0)= (-1)^{mr+s}2^{s-\frac32+mp(s-\frac12+m)}(2\pi)^{-\frac{m+1}{2}}\Gamma(s+\frac12)
\prod_{j=0}^{m-1}\frac{\Gamma(jp+r+\frac32)\Gamma(jp+p+2)}{\Gamma(p+1)}.
\end{multline*}
\end{corollary}

\section{Concluding remarks}\label{concl}

We have seen that the story of Macdonald-Mehta  integrals goes beyond the Coxeter premise
and can be extended at least to all known Baker-Akhiezer configurations.
We have used Dotsenko-Fateev formulas to extend it further to the deformed root systems \cite{SV},
related to the series of basic classical Lie superalgebras, but the exceptional cases are still to be studied.

In particular, for the exceptional Lie superalgebra of type $D(2,1,\lambda)$ we have the following 3D integral depending on 3 parameters $\lambda_1, \lambda_2, \lambda_3:$
\begin{equation}
I(\lambda)=\int \frac{x_1^{-2m_1}x_2^{-2m_2}x_3^{-2m_3}e^{-\frac{1}{\lambda_1}x_1^2-\frac{1}{\lambda_2}x_2^2-\frac{1}{\lambda_3}x_3^2} \, dx}{(x_1+x_2+x_3)^2(x_1-x_2+x_3)^2(x_1+x_2-x_3)^2(x_1-x_2-x_3)^2},
\end{equation}
where
$$m_1=\frac{\lambda_2+\lambda_3-\lambda_1}{2\lambda_1},\,m_2=\frac{\lambda_3+\lambda_1-\lambda_2}{2\lambda_2},\,m_3=\frac{\lambda_1+\lambda_2-\lambda_3}{2\lambda_3}.$$

The link of the deformed trigonometric Calogero-Moser systems \cite{SV} with conformal field theory via Dotsenko-Fateev work \cite{DF} looks mysterious to us and needs a better understanding (see \cite{SSAFR} for some results in the non-deformed case).

As it was already mentioned the initial motivation for this work came from the theory of quasi-invariant and $m$-harmonic polynomials \cite{FV02, EG02}.
We hope that our results will help to answer some remaining questions in this theory, for example, about the signature of the canonical bilinear form on the space of $m$-harmonic polynomials.  As it follows from our results (see e.g. Proposition 4.2) this form is sometime bound to be indefinite in contrast to the case $m=0.$

In this relation we would like to mention here important recent papers by Grinevich and Novikov \cite{GN09, GN11}, who developed a spectral theory for the one-dimensional Schr\"odinger operators with singular potentials in the spaces with indefinite bilinear form. Their ideas may help to make a progress in multidimensional case as well.

Another natural link to explore is with the paper \cite{FelV09} by Felder and one of the authors, where an integral  formula for the full Baker-Akhiezer function in both ${\mathcal A}_m$ and deformed ${\mathcal A}_m(p)$ cases as an iterated residue and Selberg-type integral were found.

\section{Acknowledgements}
Two of us (MVF and APV) are very grateful to Giovanni Felder for illuminating discussions during our visit to ETH Zurich in Spring 2003, when the first calculations in this direction had been done. We also would like to thank Pavel Etingof for useful discussions in November 2001.

We all are grateful to the Hausdorff Institute, Bonn for the hospitality during the trimester programme "Integrability in Geometry and Mathematical Physics" in January-April 2012, when this paper was drafted.

APV is grateful to F. Smirnov for useful discussions during the Integrability programme at the Simons Center for Geometry and Physics  at Stony Brook and to the Simons Center for the hospitality in September 2012.

The work of APV was partially supported by the EPSRC (grant EP/J00488X/1). MVF acknowledges support  from Royal Society/RFBR joint project JP101196/11-01-92612.

\bibliographystyle{amsalpha}

\end{document}